\documentclass[10pt,final,doublecolumn]{IEEEtran}
\usepackage{amsmath}
\usepackage{amssymb}
\usepackage{amsfonts} 
\usepackage[short]{optidef}
\allowdisplaybreaks[4]
\usepackage[subtle]{savetrees}

\usepackage{amsthm}
\usepackage{enumitem}

\usepackage{graphicx}
\usepackage{float}  
\usepackage{tikz,lipsum}
\usepackage[labelformat=simple]{subcaption}
\DeclareCaptionLabelSeparator{periodspace}{.\quad}
\usepackage{tabularx}
\usepackage{comment}
\usepackage{caption}
\usepackage{subcaption}

\usepackage{pdfpages}   
\usepackage{epstopdf}
\usepackage{bbm}
\usepackage{cite}
\usepackage{url} 

\usepackage[ruled,vlined]{algorithm2e}  
\usepackage[normalem]{ulem}

\usepackage{mathtools}
\usepackage{setspace}

\def\nb0{{\mathbf{0}}}
\def\nb1{{\mathbf{1}}}

\def\nbbC{{\mathbb{C}}}

\def\nbbR{{\mathbb{R}}}

\def\nbbZ{{\mathbb{Z}}}

\newtheorem{nrem}{Remark}
\newtheorem{prop}{Proposition}

\newtheorem{eg}{Example}

\usepackage{multirow}
\usepackage{lipsum}

\begin{document}
\title{\date{} Full-Space Wireless Sensing \\ Enabled by Multi-Sector Intelligent Surfaces}
\author{Yumeng Zhang,~\IEEEmembership{Student Member,~IEEE},~Xiaodan Shao,~\IEEEmembership{Member,~IEEE},~ Hongyu Li,~\IEEEmembership{Student Member,~IEEE},~\\Bruno Clerckx,~\IEEEmembership{Fellow,~IEEE},~Rui Zhang,~\IEEEmembership{Fellow,~IEEE}
\thanks{This work has been partially supported by UKRI grant EP/Y004086/1, EP/X040569/1, EP/Y037197/1, EP/X04047X/1, EP/Y037243/1.}
\thanks{Yumeng Zhang, Hongyu Li and Bruno Clerckx are with the Department of Electrical and Electronic Engineering, Imperial College London, London SW7 2AZ, U.K. (e-mail:
\{yumeng.zhang19,c.li21, b.clerckx\}@imperial.ac.uk).}
\thanks{X. Shao is with the Institute for Digital Communications, Friedrich-Alexander-University Erlangen-Nurnberg (FAU), 91054
Erlangen, Germany (email:xiaodan.shao@fau.de).}
\thanks{R. Zhang is with School of Science and Engineering, Shenzhen Research Institute of Big Data, The Chinese University of Hong Kong, Shenzhen, Guangdong 518172, China (e-mail: rzhang@cuhk.edu.cn). He is also with the Department of Electrical and Computer Engineering, National University of Singapore, Singapore 117583 (e-mail: elezhang@nus.edu.sg).
}}
\maketitle

\begin{abstract}
The multi-sector intelligent surface (IS), benefiting from a smarter wave manipulation capability, has been shown to enhance channel gain and offer full-space coverage in communications. However, the benefits of multi-sector IS in wireless sensing remain unexplored. 
  {This paper introduces the application of multi-sector IS for wireless sensing/localization. Specifically, we propose a new self-sensing system, where an active source controller uses the multi-sector IS geometry to reflect/scatter the emitted signals towards the entire space, thereby achieving full-space coverage for wireless sensing. Additionally, dedicated sensors are installed aligned with the IS elements at each sector, which collect echo signals from the target and cooperate to sense the target angle.} In this context, we develop a maximum likelihood estimator of the target angle for the proposed multi-sector IS self-sensing system, along with the corresponding theoretical limits defined by the Cramér-Rao Bound.  The analysis reveals that the advantages of the multi-sector IS self-sensing system stem from two aspects: enhancing the probing power on targets (thereby improving power efficiency) and increasing the rate of target angle (thereby enhancing the transceiver's sensitivity to target angles). Finally, our analysis and simulations confirm that the multi-sector IS self-sensing system, particularly the $4$-sector architecture, achieves full-space sensing capability beyond the single-sector IS configuration.  Furthermore, similarly to communications, {employing directive antenna patterns on each sector's IS elements and sensors} significantly enhances sensing capabilities. This enhancement originates from both aspects of improved power efficiency and target angle sensitivity, with the former also being observed in communications while the latter being unique in sensing. 
\end{abstract}


\section{Introduction}
The future wireless networks are expected to hold an increasing number of high-demand applications, such as autonomous driving and the Internet of Things \cite{letaief2019roadmap,hassanien2019dual,tataria20216g}. This growth poses greater challenges to spectrum resources and service qualities. One promising solution is the emerging technique of integrated sensing and communications (ISAC) \cite{liu2022integrated,cui2021integrating,ma2020joint}. 
ISAC enhances the spectrum utilization efficiency by realizing dual functions (sensing and communications) with shared hardware, platform, and radio resources, thereby reducing costs and optimizing resource use \cite{liu2022survey,ma2020joint}. Furthermore, exploiting the synergy between the dual functions can enhance the overall performance of ISAC, i.e., sensing capabilities enhance environmental awareness for better communication strategies, which in turn support ultra-high data rates, reliability, and ultra-low latency for communications with reduced complexity \cite{zhang2022time}.  {However, in practical scenarios with complex radio propagation environments, the performance of ISAC may significantly degrade when transmission links are obstructed by obstacles.}

{One promising solution is employing the advanced technique of intelligent surfaces (ISs), which can enhance both sensing and communication performance by reconfiguring radio environments \cite{kaina2014shaping}.} Specifically, ISs adjust the phases and/or amplitudes of the impinging signals so that the reflected signals work constructively, or build a virtual line-of-sight (LoS) path in the presence of physical obstacles \cite{liaskos2018new,wu2021intelligent,di2020smart}.  {While there has been extensive research investigating IS-enhanced wireless communications \cite{guo2020weighted,wu2019intelligent}, only a few prior works have studied on IS-enhanced sensing \cite{elzanaty2021reconfigurable,wang2023target,shao2022target} or IS-enhanced ISAC \cite{meng2023ris,yu2022location}. Among these works, one notable research is in \cite{shao2022target}, where a new type of IS-aided sensing system, called IS self-sensing, is proposed. In IS self-sensing, an active source controller serves as a transmitter to send well-adjusted probing signals towards the IS elements. This allows the illuminated IS elements to autonomously radiate (via the IS controller) and receive (via sensors) sensing signals for target localization, eliminating the dependence on the sensing signals from a dedicated base station (BS). }

{The limitation of the aforementioned IS self-sensing system is its serving range only covers the half-space facing the IS, since the incident signals are purely reflected to the same side of IS.}  However, in practice, targets can be located anywhere in the full-space. Therefore, it is crucial to leverage multiple ISs to cover the full-space and collaborate in aiding target parameter estimation.   {To achieve full-space coverage, a notable effort has been made in communications, where a multi-sector IS geometry is proposed \cite{li2023beyond}, as depicted in Fig. \ref{fig:ISs_configuration}. Therein, $L~(L\geq 2)$ antenna arrays of IS elements are arranged along the sides of a uniform prism. The $L$-sector IS achieves full-space coverage when each of the IS is activated for radiation.} \cite{li2023beyond} analytically shows that the multi-sector IS, besides the advantages of full-space coverage, also enhances the received power when employing directive antenna patterns, which boosts the communication performance further \cite{tang2020wireless,li2023beyond,shao20246d_DIS,shao20246d}.

\begin{figure}
    \centering
    \includegraphics[width=3.0in]{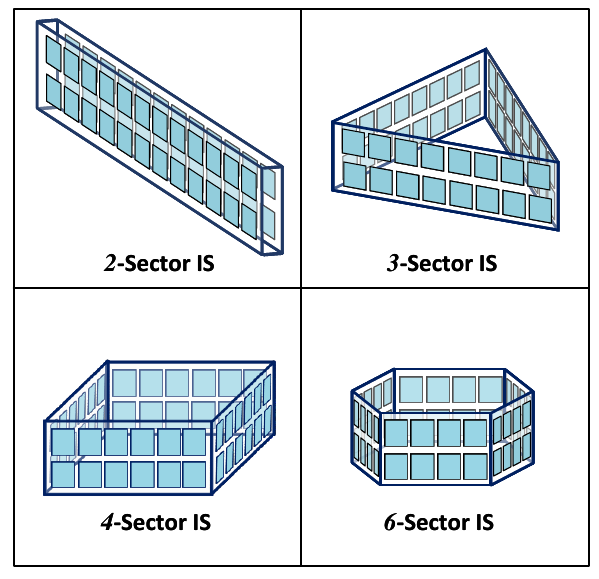}
    \caption{Examples of multi-sector IS configurations, for $L=2,~3,~4,~6$ with $L$ being the number of ISs.}
    \label{fig:ISs_configuration}
\end{figure}

While the advantages of multi-sector IS have been studied in communications \cite{li2023beyond}, its benefits in wireless sensing remain unexplored.   {Apart from the power benefit, the exploration of multi-sector IS for wireless sensing is crucial because the multi-sector IS provides flexible geometry configurations (i.e., different $L$ in Fig. \ref{fig:ISs_configuration}), which plays a vital role in the accuracy of target angle estimation \cite{karimi1996manifold,6193133}.} Indeed, the importance of transmitters'/receivers' geometry in wireless sensing has been emphasized in the traditional active multiple input multiple output/phased-array radar \cite{6193133,vasanelli2017optimization}. {Specifically, the uniform planar array radar (i.e., corresponding to the conventional IS architecture in \cite{shao2022target}) achieves high precision in sensing a target that is directly facing the antenna array, but suffers low precision when the target is in line with the extended plane of the antenna array \cite{karimi1996manifold}}. In contrast, a circular array radar offers uniform sensing performance across the full-space \cite{ioannides2005uniform}. Therefore, determining the optimal geometry (i.e., the optimal $L$) of multi-sector IS to balance communications, sensing, and implementation costs presents a compelling area for further research in ISAC. 

Towards that, this paper aims to investigate the optimal multi-sector IS geometry for wireless sensing with full-space coverage as a solid foundation for the future extension to ISAC. The main contributions of this paper are summarized as follows.

\begin{figure}
\centering\includegraphics[width=3.0in]{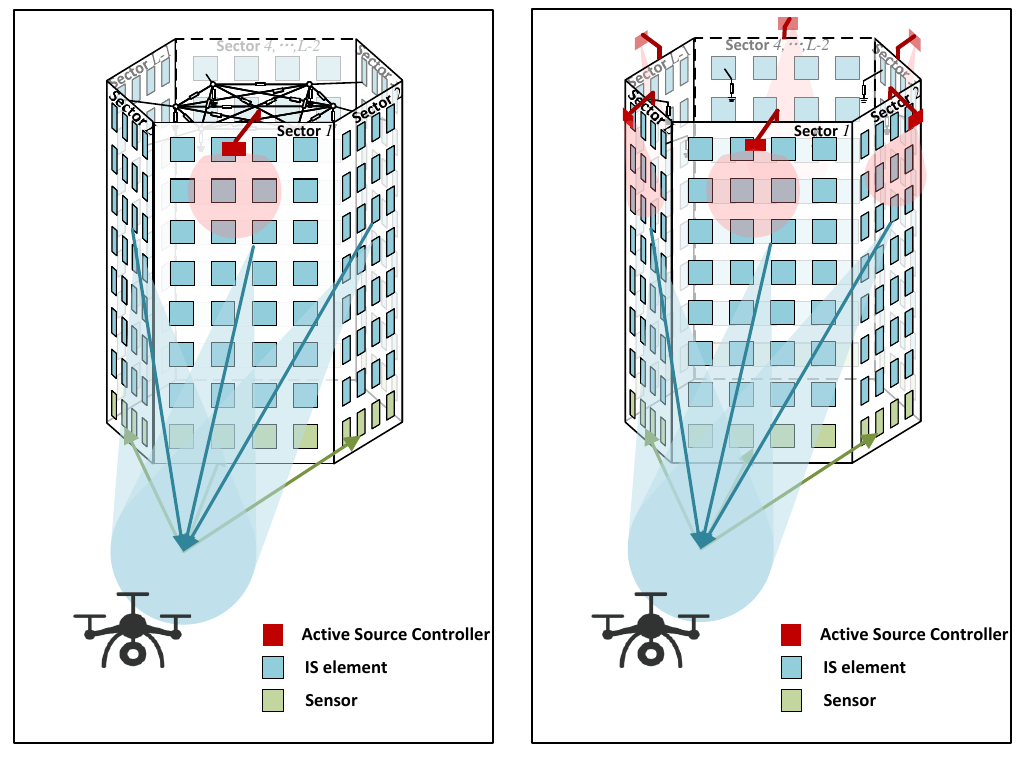}
    \caption{Two implementations of the proposed multi-sector IS self-sensing system, namely,  multi-sector beyond-diagonal IS (left), and multi-sector conventional IS (right).}
    \label{fig1_2_structure}
\end{figure}

\begin{enumerate}
    \item   {We propose a novel multi-sector IS self-sensing system for full-space coverage sensing. Specifically, we have two implementations of the multi-sector IS self-sensing system, as shown in Fig. \ref{fig1_2_structure}, where each sector is installed with dedicated IS elements (for signal radiation) and sensors (for signal collection). In the first case (Fig. \ref{fig1_2_structure}, left), an active source controller is installed on one sector and probes signals towards that sector. In addition, the IS elements across the $L~(L\geq 2)$ sectors are connected through a reconfigurable impedance network. The network connection enables the probing signal to be simultaneously reflected by the sector with the active source controller and scattered through the remaining sectors, which facilitates full-space radiation. In the second case  (Fig. \ref{fig1_2_structure}, right), $L$ active source controllers are installed, with one active source controller per sector. In this scenario, the IS elements at each sector are independently illuminated by their aligned active source controller and reflect the impinging signals for radiation.   For both implementations in Fig. \ref{fig1_2_structure}, the implemented sensors across the $L$ sectors work collaboratively to estimate the target angle based on their collected target echo signals.} 
 
    \item   {A joint  {maximum-likelihood (ML)} estimator of the target angle is derived for the proposed multi-sector IS self-sensing system. The Cramér-Rao bound (CRB) is also analytically derived as the lower bound of the empirical mean squared error (MSE) for estimating the target angle. By analyzing the CRB, we show that the estimation performance is mainly affected by two fundamental factors, i.e., \textit{the probing power on the target} and \textit{the squared rate of target angle}. Additionally, we specify the CRB and compare the MSEs for both half-space isotropic and half-space directive antenna patterns through a thorough analysis.} 
    
    \item {Simulation and analytic results show that the proposed multi-sector IS self-sensing system achieves full-space coverage, thus surpassing the capabilities of the single-sector IS configuration. Notably, when comparing among $L=2,~3,~4$, the configuration of $4$-sector IS achieves the best overall sensing performance, and particularly outperforms the conventional simultaneously transmitting and reflecting
reconfigurable surface (STARS) architecture in \cite{wang2023simultaneously} (i.e., the special case of the proposed multi-sector IS self-sensing system with $L=2$).  Moreover, employing half-space directive antenna patterns on each sector's IS elements and sensors significantly improves the sensing performance compared with that of the half-space isotropic antenna pattern.} 
\end{enumerate}

 The rest of this paper is organized as follows.   In Section \ref{sec2: sys_model}, we introduce the system and signal model of a multi-sector IS self-sensing system for full-space coverage. Section \ref{se3: analysis} derives the ML estimator and the CRB for estimating the target angle.  Section \ref{sec4: simulation} provides simulation results for verification. Section \ref{sec5:conclusion} concludes the paper.

\textit{Notation:}  
 Throughout the paper, matrices and vectors are respectively denoted in bold upper case and bold lower case. $\nbbR$, $\nbbC$ and $\nbbZ$ denote the set of real numbers, complex numbers and integers, respectively. For $x \in \nbbC$, $\mathfrak{R}\{x\}$, $\mathfrak{I}\{x\}$, $\measuredangle x$ and $|{x}|$ represents the real part, the imaginary part, the phase and the magnitude of ${x}$, respectively. For a vector (matrix) $\mathbf{x}$ ($\mathbf{X}$), $\|\mathbf{x}\|$ ($\|\mathbf{X}\|$) represents its $l_2$ (Frobenius) norm, $x_n$ (or $\left[\mathbf{x}\right]_n$) refers to the $n^{\mathrm{th}}$ entry of vector $\mathbf{x}$, and ${X}_{k,m}$ (or $\left[\mathbf{X}\right]_{k,m}$) is the $\left(k^{\mathrm{th}},~m^{\mathrm{th}}\right)$ entry of matrix $\mathbf{X}$. $\mathrm{diag}\{\mathbf{x}\}$($\mathrm{diag}\{\mathbf{X}\}$) represents a diagonal matrix (vector) with its $n^{\mathrm{th}}$ diagonal entry ($n^{\mathrm{th}}$ entry) equal to the $n^{\mathrm{th}}$ entry (diagonal entry) of vector $\mathbf{x}$ (matrix $\mathbf{X}$).   $\mathbf{I}_{N_t}$ denotes an  $N_t\times N_t$ identity matrix and $\mathbf{1}_K/\mathbf{0}_K$ is a column all-one/all-zero vector with dimension $K$. $(\cdot)^H$, $(\cdot)^*$ and $(\cdot)^T$ represent the Hermitian, the conjugate and the transpose operation respectively. $\mathbbm{1}(\cdot)$ is the indicator function. For a random variable $X$, $\mathbb{E}_X\left\lbrace f(X)\right\rbrace$ represents the expectation of function $f(X)$ averaged over $X$. {$\overset{N\uparrow}{\approx}$ and $\overset{N\uparrow}{=}$ respectively mean approximately equal and equal as $N$ grows large.}  \textit{All the positions/distances are normalized by half-wavelength if not being specified.}

\section{System Model}\label{sec2: sys_model}
This section first describes the two implementations of the multi-sector IS self-sensing system in Section \ref{sec2-1:configurations}. Later, Section \ref{sec2-2:Geo} sets up the Cartesian coordinate system (CCS) of the multi-sector IS self-sensing system and expresses the positions of the IS elements/sensors, by focusing on one implementation in Section \ref{sec2-1:configurations}. Section \ref{sec2-3:configurations} then derives the mathematical expressions of the radiated signal and its radar echoes in the multi-sector IS self-sensing system.

\subsection{Two implementations of the proposed multi-sector IS self-sensing system}\label{sec2-1:configurations}

  {As shown in Fig. \ref{fig1_2_structure}, we consider a multi-sector IS self-sensing system, where an $L$-sector IS configuration in \cite{li2023beyond} is combined with the self-sensing structure in \cite{shao2022target} to sense the angle of a point-target that can be located across the full-space. Specifically, we install a total number of $N_{\mathrm{I}}$ IS elements, with $M_{\mathrm{I}}=N_{\mathrm{I}}/L$ IS elements per sector. Moreover, to estimate the target angle, we install a total number of $N_{\mathrm{S}}$ sensors aligned with the IS elements, with $M_{\mathrm{S}}=N_{\mathrm{S}}/L$ sensors per sector to receive the echoed signals from the target. We make a joint estimation of the target angle after collecting all the echoed signals from all sensors. For simplicity, we focus on the target's azimuth angle estimation in this paper, while the results can be extended to estimate the elevation angle as well.}

  {Specifically, Fig. \ref{fig1_2_structure} depicts two implementations of the multi-sector IS self-sensing system, which are distinguished by the way of activating the IS elements across different sectors. In the first implementation as shown in Fig. \ref{fig1_2_structure} (left), we install only one active source controller, which probes signals towards the sector it faces (named by sector $1$) and only illuminates the IS elements at sector $1$. The impinging signal from the active source controller is then partially reflected by the IS elements at sector $1$ and partially transmitted to and scattered from the other sectors through the reconfigurable impedance networks. This allows the signal to be radiated from all the sectors and hence enables full-space sensing for the target angle estimation.  The interconnection network across sectors is realized by an emerging IS technique, i.e., beyond diagonal IS \footnote{The term, 'beyond diagonal', characterizes the mathematical expressions for the phase-shift matrix of IS elements when introducing reconfigurable impedance networks to connect between IS elements. Compared with the conventional IS where IS elements are independent of each other, introducing circuit connections turns the phase-shift matrix from a diagonal matrix to a non-diagonal matrix.} in \cite{shen2021modeling,li2022beyond,li2023reconfigurable}. Hence, we name the implementation in Fig. \ref{fig1_2_structure} (left)  by multi-sector beyond diagonal IS.  In contrast, in the second implementation as shown in Fig. \ref{fig1_2_structure} (right), each sector is equipped with its own active source controller, and hence IS elements across different sectors are illuminated independently by the probing signal from their corresponding active source controller. Then, all IS elements reflect their impinging signal for full-space radiation. In this context, the second implementation resembles deploying multiple conventional IS self-sensing structures (with the diagonal phase-shift matrix) in \cite{shao2022target} in a multi-sector geometry in \cite{li2023beyond}, hence named by multi-sector conventional IS as shown in Fig. \ref{fig1_2_structure} (right).}
 
  {The two implementations have pros and cons from the following perspectives. On one hand, the first implementation involves more circuit complexity due to interconnected ISs but requires only one active source controller, which simplifies the synchronization control. On the other hand, the second implementation features a less complicated circuit but requires multiple active source controllers and perfect synchronization for effective cooperation between sectors. In this paper, our assumptions for the target angle estimation hold for both implementations, and hence we focus on the first implementation without loss of generality. }

\subsection{Geometry constellation}\label{sec2-2:Geo}
\begin{figure*}[ht]
    \centering
\includegraphics[scale=0.7]{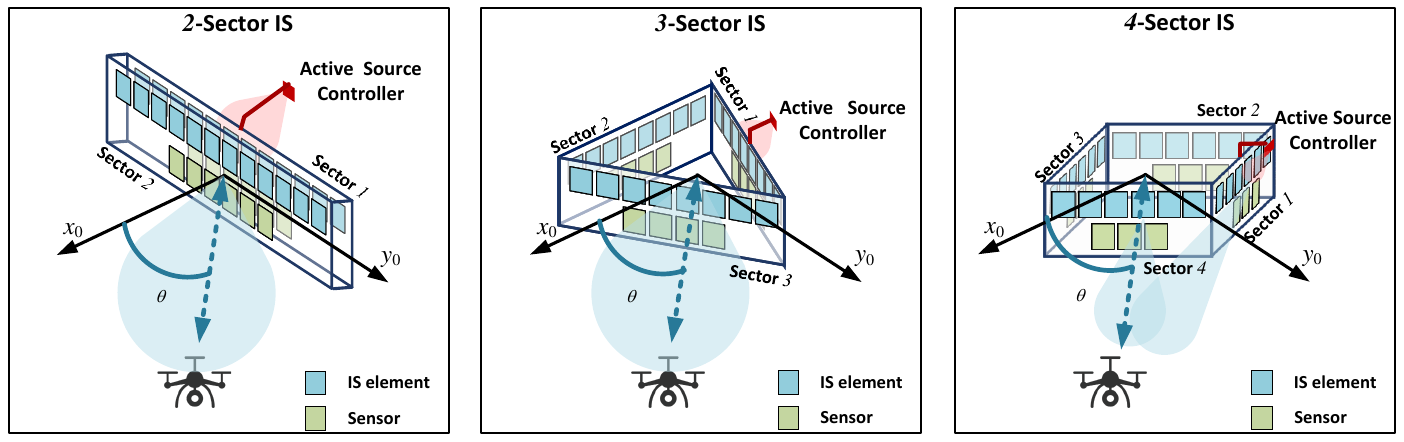}
    \caption{ {The model of the multi-sector IS self-sensing system under \textit{global CCS} for $L=2,~3,~4$ from left to right, assuming $N_\mathrm{I}=24$ and $N_\mathrm{S}=12$.}}
    \label{fig_system_model}
\end{figure*}

\begin{figure*}[ht]
    \centering    \includegraphics[scale=0.7]{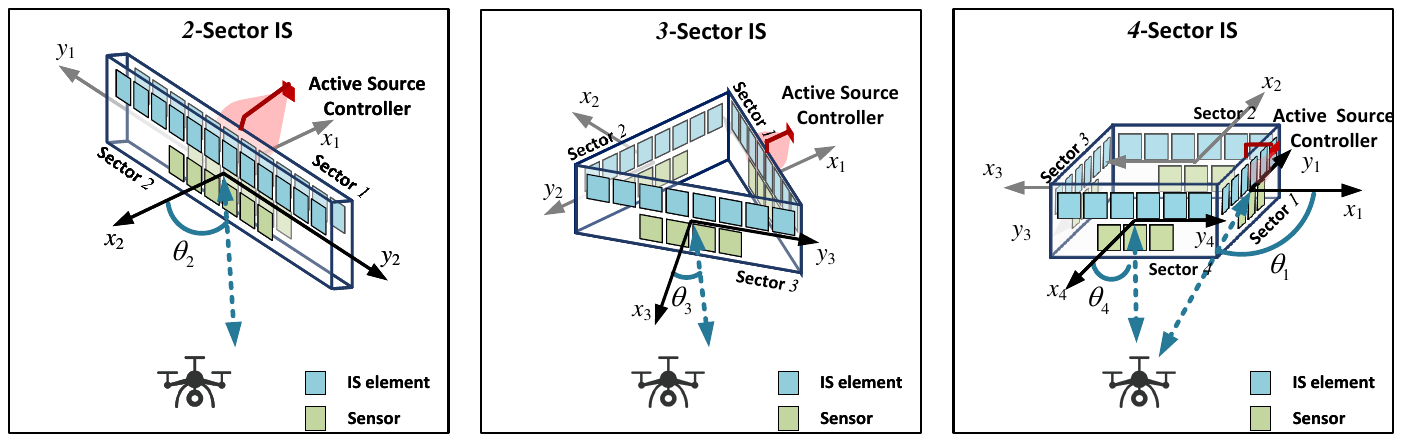}
    \caption{ {The model of the multi-sector IS self-sensing system under \textit{local CCS} for $L=2,~3,~4$ from left to right, assuming $N_\mathrm{I}=24$ and $N_\mathrm{S}=12$. Herein, the target angle to be estimated  in the G-CCS, $\theta$, becomes $\theta_l$ adaptive to each local CCS.}}
    \label{fig_system_model2}
\end{figure*}

In this sub-section, we express the positions of the target, IS elements and sensors in the multi-sector IS self-sensing system, with a particular focus on the configurations for $L=2,~3,~4$\footnote{ For practical considerations, the forthcoming comparisons between $L$ are limited to $L=2,~3,~4$, as expanding to a larger number of sectors yields significantly complicating hardware implementation.}, as shown in Fig. \ref{fig_system_model}. For simplicity,  we assume the IS elements/sensors at each sector are arranged in a uniform linear array (ULA) with half-wavelength spacing, and the IS elements and sensors are located at the same elevation level. In this context, we establish the global CCS (G-CCS) in Fig. \ref{fig_system_model}, denoted by $x_0-y_0$, with the multi-sector IS's center serving as the origin and the $y_0$-axis extending towards the edge of sector $1$. In the G-CCS, starting from sector $1$, the other sectors are named sequentially in an anti-clockwise direction. For clarification, we also establish local CCS (L-CCS) for each sector under each configuration, denoted by $x_l-y_l,~l=1,~2,~\cdots,~L$, with the sector center serving as the origin,  as shown in Fig. \ref{fig_system_model2}. Apart from the origin,  $x_l-y_l$ is established with its $y_l$-axis extending along the ULA of IS elements at the $l^{\text{th}}$ sector and its $x_l$-axis extending outward from the $l^{\text{th}}$ sector plane, as illustrated in Fig. \ref{fig_system_model2}.

First, denote the position of the unknown target  in the G-CCS $x_0-y_0$ as $\mathbf{p}_{T}\in\mathbb{R}^2$, which is given by 
\begin{align}
    \label{eq_target}
\mathbf{p}_\mathrm{T}=&\rho\mathbf{u}\left(\theta \right)=\rho[\cos(\theta),~\sin(\theta)]^T,
\end{align}
where $\theta$ denotes the target azimuth angle to be estimated and $\|\mathbf{p}_{T}\|=\rho$ denotes the distance between the target and the $x_0-y_0$'s origin. 

 {Then, we construct the positions of the IS elements/sensors at each sector in the G-CCS $x_0-y_0$, which can be converted from their positions in the L-CCS  $x_l-y_l$ by utilizing the position transformations between $x_l-y_l  $ and $x_0-y_0 $.  {The positions of the IS elements/sensors at the $l^{\text{th}}$ sector in $x_l-y_l$ are given by}
\begin{align}\label{eq_local_CCS_censor}
\mathbf{P}^{\mathrm{L}}_\mathrm{I/S}=&\left[\mathbf{p}^{\mathrm{L}}_{\mathrm{I/S},~x},~\mathbf{p}^{\mathrm{L}}_{\mathrm{I/S},~y}\right]^{T}\in\mathbb{R}^{2\times {M_{\mathrm{I/S}}}},
\end{align} 
where we have $\mathbf{p}^{\mathrm{L}}_{\mathrm{I/S},~x}=\mathbf{0}_{M_{\mathrm{I/S}}}\in\mathbb{R}^{{M_{\mathrm{I/S}}}}$, and  $\mathbf{p}^{\mathrm{L}}_{\mathrm{I/S},~y}=\left[-\frac{M_{\mathrm{I/S}}-1}{2},~\cdots,~-\frac{1}{2},~\frac{1}{2},~\cdots,~\frac{M_{\mathrm{I/S}}-1}{2}\right]^T\in\mathbb{R}^{{M_{\mathrm{I/S}}}}$.

Then, to build up the position transformations between $x_l-y_l$ and $x_0-y_0 $, we introduce the rotation matrix as follows,
\begin{subequations}\label{eq_rotation}
  \begin{align}
\mathbf{Q}_{l}=&\begin{bmatrix}\cos{\phi_l},&-\sin{\phi_l}\\\sin{\phi_l},&\cos{\phi_l} \end{bmatrix},\\   \mathbf{Q}^{\mathrm{L}}_{l}=&\begin{bmatrix}\cos{\phi_l},&\sin{\phi_l}\\-\sin{\phi_l},&\cos{\phi_l} \end{bmatrix},
\end{align}  
\end{subequations}
where $\mathbf{Q}_{l}$ and $\mathbf{Q}^{\mathrm{L}}_{l}$ are respectively the rotation matrix from $x_0-y_0$ to $x_l-y_l$ and  from $x_l-y_l$ to $x_0-y_0$. $\phi_l$ is the rotation angle from $x_0-y_0$ to $x_l-y_l$, i.e., $\phi_l=\pi/2+(2l-1)\pi/L$. 

Combining \eqref{eq_local_CCS_censor} and \eqref{eq_rotation}, the positions of the  IS elements/sensors at the $l^{\text{th}}$ sector in the G-CCS $x_0-y_0$ are expressed as
\begin{subequations}
    \begin{align}
\label{eq_global_sensor}
\mathbf{P}_{\mathrm{I/S},~l}=&{\mathbf{Q}^{\mathrm{L}}_{l}}^T\left(\mathbf{P}^{\mathrm{L}}_\mathrm{I/S}-\mathbf{p}^{\mathrm{L}}_{0,~l}  {\mathbf{1}^T}\right)\in\mathbb{R}^{2\times M_\mathrm{I/S}},\\\label{eq_global_sensor1}
\mbox{with}~
\mathbf{p}^{\mathrm{L}}_{0,~l}=&{\mathbf{Q}_{l}}^T\left(-\mathbf{p}_{\mathrm{c},~l}\right)=\left[-\frac{M_\mathrm{I/S}}{2\tan\left(\pi/L\right)},~0\right]^T\in\mathbb{R}^{2},\\\label{eq_global_sensor2}
\mathbf{p}_{\mathrm{c},~l}=&\frac{M}{2\tan{\left(\pi/L\right)}}\left[\cos{\phi_l},~\sin{\phi_l}\right]^T\in\mathbb{R}^{2},
\end{align}
\end{subequations}
where $\mathbf{p}^{\mathrm{L}}_{0,~l}$ denotes the position of $x_0-y_0$'s origin in $x_l-y_l$. $\mathbf{p}^{\mathrm{L}}_{0,~l}$ is obtained by the position conversion from $x_l-y_l$ to $x_0-y_0$, as expressed in \eqref
{eq_global_sensor1}. The position conversion therein also requires the position of $x_l-y_l$'s origin in $x_0-y_0$, which is denoted by $\mathbf{p}_{\mathrm{c},~l}$ in \eqref
{eq_global_sensor2}.  In \eqref
{eq_global_sensor2}, 
 ${2\tan\left(\pi/L\right)}/M$ with $M\triangleq\max \left(M_\mathrm{I},~M_\mathrm{S}\right)$ is the distance from $x_0-y_0$'s origin to $ \mathbf{p}_{\mathrm{c},~l}$.

\subsection{Signal model}\label{sec2-3:configurations}
We assume a narrow-band system, where no direct link exists between the active source controller and the target, as shown in Fig. \ref{fig_system_model}. Moreover, the active source controller consecutively sends probing signals over $Q$ snapshots directively towards IS elements at sector $1$. For each snapshot, the impinging signal is simultaneously reflected by the IS elements at sector $1$ and scattered through the IS elements at the remaining sectors with well-designed phase shifts. Then, the signal is radiated in the full-space from all sectors, and the echoed signals by the target are received by the sensors.  For target angle estimation, we collect the echoed signals from all sensors over the whole $Q$ snapshots and conduct a joint ML estimation. The whole transmission is modelled as follows in detail.

First, denote $s_q,~q=0,~\cdots,~Q-1,$ as the probing signal from the active source controller at snapshot $q$. For each snapshot $q$, the signal $s_q$ undergoes the channel from the active source controller to the IS elements at sector $1$, which is assumed to be known and is modeled as a far-field LoS channel. The channel of active source controller$\rightarrow$IS elements at sector $1$, denoted by $\mathbf{g}\in\mathbb{C}^{M_{\mathrm{I}}}$, can then be given by
\begin{align}
    \label{eq_channel_g}
   \mathbf{g}=& \alpha_\mathrm{g}\mathbf{a}_\mathrm{g},
\end{align}
where $\alpha_\mathrm{g}=\sqrt{\frac{\lambda^2G_\mathrm{T}
G_\mathrm{R}}{16\pi^2d^2_{\mathrm{CI}}}}\exp\left\lbrace \frac{j2d_{\mathrm{CI}}\pi}{\lambda}\right\rbrace$ denotes the complex-valued path gain of the active source controller$\rightarrow$IS elements at sector $1$  channel. In the above, $d_{\mathrm{CI}}$ represents the distance between the active source controller and the $1^{\mathrm{st}}$ sector, $\lambda$ is the carrier wavelength, $G_\mathrm{T}$ is the gain of the antenna at the active source controller, and $G_\mathrm{R}$ is the gain of the antenna at each IS element. Notice that in the following, we assume $G_\mathrm{T}G_\mathrm{R}=\beta /M_\mathrm{I}$ to reflect that the active source controller is directing only towards the $M_\mathrm{I}$ IS elements at sector $1$\footnote{The larger the value of  $M_\mathrm{I}$, the larger the antenna aperture of IS of sector $1$, and the weaker the antenna directivity (hence the antenna gain) of the active source controller.}. Without loss of generality, we choose proper $\beta$ and $d_{\mathrm{CI}}$ to make the path loss $\alpha_\mathrm{g}=\sqrt{1/{M_{\mathrm{I}}}}$. In addition, in \eqref{eq_channel_g}, $\mathbf{a}_\mathrm{g}=\exp\left\lbrace j\pi\mathbf{p}^{\mathrm{L}}_{\mathrm{I},~y}\sin\left(\zeta\right)\right\rbrace\in\mathbb{C}^{M_\mathrm{I}}$ is the steering vector of IS elements at sector $1$, with $\zeta$ being the angle between the active source controller and the sector $1$'s center.

In this context, the radiated signal from the IS elements at the $l^\text{th}$ sector at the snapshot $q$ is defined as 
\begin{align}\label{eq_tran_signal}
    \mathbf{x}_{l,q}=&\sqrt{P^{\mathrm{tr}}}\bar{\mathbf{\Phi}}_{l,q} \mathbf{g} s_q\in\mathbb{C}^{M_\mathrm{I}},
\end{align}
where  $P^{\mathrm{tr}}$ is the transmit power at the active controller. $\bar{\mathbf{\Phi}}_{l,q}$ is a diagonal matrix that characterizes {the IS elements' phase shift from the $1^\text{st}$ sector to the $l^\text{th}$ sector at time $q$}. {We construct $\bar{\mathbf{\Phi}}_{l,q}$ such that, after combining with channel $\mathbf{g}$, the emitted signal $\mathbf{x}_{l,q}$ from sector $l$ at time $q$ is coincident with the $(q \mod ~M_\mathrm{I})^\mathrm{th}$ column of the $M_\mathrm{I}-$DFT matrix\footnote{We make this choice because literature has demonstrated that a DFT codebook is optimal for target sensing when there is no prior knowledge of the target \cite{shao2022target}.}, i.e., $\bar{\mathbf{\Phi}}_{l,q}=\sqrt{\frac{1}{L}}\mathrm{diag}\left\lbrace \mathbf{f}_{q}\right\rbrace\mathrm{diag}\left\lbrace \mathbf{a}_\mathrm{g}^*\right\rbrace$. In this context, we satisfy the rule of lossless passive IS elements, i.e., $\sum_{l}\bar{\mathbf{\Phi}}_{l,q}^H\bar{\mathbf{\Phi}}_{l,q}=\mathbf{I}_{M_\mathrm{I}}$. Therein, $ \mathbf{f}_{q}$ is the $(q \mod ~M_\mathrm{I})^\mathrm{th}$ column of the $M_\mathrm{I}-$DFT matrix.  
In this paper, we set  $Q=N_\mathrm{I}$ for simplicity such that the overall codebook has a periodicity of $L$.}

 { In this context, the radiated signal $\mathbf{x}_{l,q}$ becomes a periodic DFT codebook in the following form,} 
\begin{align}  
     \label{eq_tran_signal2}
 \mathbf{x}_{l,q}=&   \sqrt{\frac{P^{\mathrm{tr}}}{N_{\mathrm{I}}}}\mathbf{f}_{q},\\
   \mbox{with}~\left[\mathbf{f}_{q}\right]_m=&  \exp\left\lbrace -\frac{2\pi mq}{M_{\mathrm{I}}}\right\rbrace,~ m=0,~\cdots,~M_{\mathrm{I}}-1,
\end{align}
where  $s_q$ is assumed to be constant, e.g., $s_q=1$.  In this context, the total power of the whole codebook is $\sum_{l,q}\|\mathbf{x}_{l,q}\|^2=QP^{\mathrm{tr}}$.

Consequently, the echoed signal from the target at the sensors at the $l^\text{th}$ sector is given by (assume the target exists in the far-field of the sector)
\begin{subequations}\label{eq_re_signal2}
   \begin{align}
     \mathbf{Y}_{l}=&\alpha 
 F_{\mathrm{S}}\left(\theta,~l \right)\mathbf{a}_l\left(\theta\right)\sum_{l'=1}^{L} F_{\mathrm{I}}\left(\theta,~l' \right)\mathbf{b}_{l'}\left(\theta \right)^H\mathbf{X}_{l'}+\mathbf{Z}_l\in\mathbb{C}^{M_{\mathrm{S}}\times Q},\\\nonumber
     \mbox{with}~ \alpha=&\sqrt{\frac{64\lambda^2}{\pi^3 \|\mathbf{p}_\mathrm{T}\|^4}}\alpha_\mathrm{T},\\ &\text{{(the path loss normalized by isotropic antenna pattern),}}\\\nonumber
      \mathbf{a}_l\left(\theta \right)=&\exp\left\lbrace -j\pi \mathbf{P}^{T}_{\mathrm{S},~l}\left(-\mathbf{u}\left(\theta \right)\right)\right\rbrace=\exp\left\lbrace j\pi \mathbf{P}^{T}_{\mathrm{S},~l}\mathbf{u}\left(\theta \right)\right\rbrace\in\mathbb{C}^{M_{\mathrm{S}}},\\
        &\text{{(the receive steering vector at sector 
 $l$),}}\\\nonumber
        \mathbf{b}_l\left(\theta\right)=&\exp\left\lbrace -j\pi \mathbf{P}^{T}_{\mathrm{I},~l}\mathbf{u}\left(\theta \right)\right\rbrace\in\mathbb{C}^{M_{\mathrm{I}}},\\
        &\text{{(the transmit steering vector at sector 
 $l$),}}\\\nonumber
          \mathbf{X}_l=&\left[\mathbf{x}_{l,0},~ \mathbf{x}_{l,2},~ \cdots,~\mathbf{x}_{l,Q-1}\right]\in\mathbb{C}^{M_{\mathrm{I}}\times Q},\\
        &\text{{(the radiated signal over 
        $Q$ snapshots at sector 
 $l$),}}
\end{align} 
\end{subequations}
where   {$F_{\mathrm{S/I}}\left(\theta,~l\right)=\sqrt{G_{\mathrm{S/I}}\left(\theta_l\right)}$, with $G_{\mathrm{S/I}}\left(\theta_l\right)$ being  {the antenna gain (towards the target) of the IS elements/sensors at sector $l$}. Therein, $\theta_l$ is the target angle in the L-CCS $x_l-y_l  $. Note 
$G_{\mathrm{S/I}}\left(\theta_l\right)=0$ indicates that the target is not illuminated by sector $l$. }  Detailed discussions about $F_{\mathrm{S/I}}\left(\theta,~l\right)$ and $G_{\mathrm{S/I}}\left(\theta_l\right)$ can be found in Section \ref{sec_AP}.   {In the path loss $\alpha$, $\alpha_\mathrm{T}$ is the complex target scattering coefficient. } $\mathbf{Z}_l\in\mathbb{C}^{M_{\mathrm{S}}\times Q}$ is the additive white Gaussian noise (AWGN) at sensors at sector $l$, with noise power being $\sigma^2$. 

Combining $\mathbf{Y}_l$ for $l=1,~\cdots,~L$ into one matrix, we have  the following expressions,
\begin{subequations}\label{eq_y_vec0}
  \begin{align}
    \nonumber\mathbf{Y}=&\left[\mathbf{Y}_{1}^T,~ \mathbf{Y}_{2}^T,~ \cdots,~\mathbf{Y}_{L}^T\right]^T\\
    =&\alpha \mathbf{F}_{\mathrm{S}}\left(\theta\right)\mathbf{a}\left(\theta\right)\mathbf{b}\left(\theta\right)^H\mathbf{F}_{\mathrm{I}}\left(\theta\right)^T\mathbf{X}+\mathbf{Z}\\
    =&\alpha \mathbf{U}\left(\theta\right)+\mathbf{Z},\\
    \mbox{with}~   \mathbf{F}_{\mathrm{I/S}}\left(\theta\right)=&\mathrm{diag}\left\lbrace F_{\mathrm{I/S}}\left(\theta,~1\right),\cdots,~F_{\mathrm{I/S}}\left(\theta,~L\right)\right\rbrace\otimes \mathbf{I}_{M_{\mathrm{I/S}}},\\
    \mathbf{a}\left(\theta\right)=&\left[\mathbf{a}_1\left(\theta\right)^T,~ \mathbf{a}_2\left(\theta \right)^T,~ \cdots,~\mathbf{a}_L\left(\theta\right)^T\right]^T\in \mathbb{C}^{N_\mathrm{S}},\\
      \mathbf{b}\left(\theta\right)=&\left[\mathbf{b}_1\left(\theta\right)^T,~ \mathbf{b}_2\left(\theta \right)^T,~ \cdots,~\mathbf{b}_L\left(\theta\right)^T\right]^T\in \mathbb{C}^{N_\mathrm{I}},\\
      \mathbf{X}=&\left[\mathbf{X}_1^T,~ \mathbf{X}_2^T,~ \cdots,~\mathbf{X}_L^T\right]^T\in \mathbb{C}^{N_{\mathrm{I}}\times Q},\\
      \mathbf{Z}=&\left[\mathbf{Z}_1^T,~ \mathbf{Z}_2^T,~ \cdots,~\mathbf{Z}_L^T\right]^T\in \mathbb{C}^{N_{\mathrm{S}}\times Q}.
\end{align}  
\end{subequations}

 The vector form of \eqref{eq_y_vec0} can be written as  {(by setting $Q=N_{\mathrm{I}}$)} 
\begin{subequations}\label{eq_y_vec}
    \begin{align}
\mathbf{y}=&\mathrm{vec}\left\lbrace{\mathbf{Y}}\right\rbrace= \alpha\boldsymbol{\mu}\left(\theta\right)+\mathbf{z}\in\mathbb{C}^{N_{\mathrm{I}}N_{\mathrm{S}}},\\
\boldsymbol{\mu}\left(\theta\right)=&\mathrm{vec}\left\lbrace \mathbf{U}\left(\theta\right) \right\rbrace=\mathbf{X}^{T}\mathbf{F}_{\mathrm{I}}\left(\theta\right)\mathbf{b}\left(\theta\right)^*\otimes \mathbf{F}_{\mathrm{S}}\left(\theta\right)\mathbf{a}\left(\theta\right)\in\mathbb{C}^{N_{\mathrm{I}}N_{\mathrm{S}}},\\
    \mathbf{z}=&\mathrm{vec}\left\lbrace \mathbf{Z}\right\rbrace\in\mathbb{C}^{N_{\mathrm{I}}N_{\mathrm{S}}}.
\end{align}
\end{subequations}

\section{Performance Analysis}\label{se3: analysis}
{ With the derived compact signal model in \eqref{eq_y_vec}, in this section, we} first present the ML estimator of the signal and the analytical lower bound of the corresponding MSE, namely the CRB,  in Section \ref{sec3-1:CRB}. Given the CRB,  we reveal its two most fundamental components, which are, respectively, the probing power on targets and the squared rate of the target angle at sensors. Then, in Section \ref{sec3-2:e_r2}, we provide an in-depth analysis of these two components of the CRB.  Finally, in Section \ref{sec3-3: scaling_summary}, we pursue the performance comparison with half-space isotropic versus half-space directive antenna patterns.

\subsection{ML estimator and the corresponding CRB}\label{sec3-1:CRB}
In this sub-section, we first derive the ML estimator of the proposed multi-sector IS self-sensing system. Given the received signal in \eqref{eq_y_vec}, the parameters to be estimated are collected into $\boldsymbol{\zeta}=\left[ \theta,~{\alpha}\right]$, which has the following log-likelihood function,  \begin{align}\label{eq_LogL}
     L\left(\mathbf{y}; ~\boldsymbol{\zeta}\right)=-\left(N_\mathrm{I}N_\mathrm{S}\right)^2\log {(\pi\sigma^2)}-\frac{1}{\sigma^2}\|\mathbf{y}-\alpha\boldsymbol{\mu}\left( \theta\right)\|^2.
 \end{align}   

Hence, the ML estimator of $\boldsymbol{\zeta}$ is given by
\begin{subequations} \label{eq_MLE_app}
    \begin{align}
\left(\hat{\theta},~\hat{{\alpha} }\right)=&\underset{\theta,~\alpha}{\arg~ \min} \|\mathbf{y}-\alpha\boldsymbol{\mu}\left( \theta\right)\|^2\\
=&\underset{\theta,~\alpha}{\arg~ \min} -2\mathfrak{R}\left\lbrace\alpha\mathbf{y}^H\boldsymbol{\mu}\left( \theta\right)\right\rbrace+\|\alpha\boldsymbol{\mu}\left( \theta\right)\|^2,
\end{align}
\end{subequations}
which gives the optimal $\hat{{\alpha}}$ as following \cite{10304548}, 
\begin{align}\label{eq_MLE_alpha}
    \hat{\alpha }=\frac{ \boldsymbol{\mu}\left( \theta\right)^H\mathbf{y}}{\|\boldsymbol{\mu}\left( \theta\right)\|^2}.
\end{align}

Taking \eqref{eq_MLE_alpha} into \eqref{eq_MLE_app}, we obtain the optimal $\hat{\theta}$ in the form of
\begin{align}
    \label{eq_MLE}
    \hat{\theta}=\underset{\theta}{\arg~ \max} \frac{\big| \boldsymbol{\mu}\left( \theta\right)^H\mathbf{y}\big|^2}{\| \boldsymbol{\mu}\left( \theta\right)\|^2}.
\end{align}

Given the ML estimator in \eqref{eq_MLE}, the corresponding  Fisher information matrix is defined as \cite{kay1993fundamentals} 
\begin{subequations}
    \begin{align}
    \label{eq_FIM}
    \mathbb{F}\left(\boldsymbol{\zeta} \right)=&\frac{2}{\sigma^2}\mathfrak{R}\left\lbrace \frac{\partial \alpha^*\boldsymbol{\mu}^H\left(\theta \right)}{\boldsymbol{\zeta}}\frac{\partial \alpha\boldsymbol{\mu}\left(\theta \right)}{\boldsymbol{\zeta}}\right\rbrace\\
    =& \frac{2}{\sigma^2}\mathfrak{R}\left\lbrace \begin{bmatrix}
|\alpha|^2\dot{\boldsymbol{\mu}}\left({\theta}\right)^H\dot{\boldsymbol{\mu}}\left({\theta}\right),&\alpha^*\dot{\boldsymbol{\mu}}\left({\theta}\right)^H\boldsymbol{\mu}\left(\theta \right)\left[1,~j \right]\\
        \left[1,~-j \right]^T\alpha \boldsymbol{\mu}\left(\theta \right)^H\dot{\boldsymbol{\mu}}\left({\theta}\right),&\boldsymbol{\mu}\left(\theta \right)^H\boldsymbol{\mu}\left(\theta \right)\mathbf{I}_2,
    \end{bmatrix}\right\rbrace, 
\end{align}
with  $\dot{\boldsymbol{\mu}}\left({\theta}\right)\triangleq\frac{\partial \boldsymbol{\mu}\left(\theta \right)}{\partial \theta}$
\end{subequations}, which yields the CRB for $\theta$ as following \cite{10304548,kay1993fundamentals},
\begin{subequations} \label{eq_CRB_conclusion}
    \begin{align}
   &\mathrm{CRB}\left( \theta\right)=\left[\mathbb{F}^{-1}\left(\boldsymbol{\zeta} \right)\right]_{1,1}\\
   =&\frac{\sigma^2}{2|\alpha|^2}\left[\dot{\boldsymbol{\mu}}\left({\theta}\right)^{H}\dot{\boldsymbol{\mu}}\left({\theta}\right)-\frac{\dot{\boldsymbol{\mu}}\left({\theta}\right)^{H}  {\boldsymbol{\mu}}\left({\theta}\right){\boldsymbol{\mu}}\left({\theta}\right)^{H}\dot{\boldsymbol{\mu}}\left({\theta}\right)}{{\boldsymbol{\mu}}\left({\theta}\right)^{H}{\boldsymbol{\mu}}\left({\theta}\right)}\right]^{-1}\\\nonumber
   =&\frac{\sigma^2}{2|\alpha|^2}\Bigg[\dot{\mathbf{b}}^T\dot{\mathbf{b}}^*{\mathbf{a}^\mathrm{p}\left(\theta\right)}^H{\mathbf{a}^\mathrm{p}\left(\theta\right)}+{\mathbf{b}}^T{\mathbf{b}}^*{ \dot{\mathbf{a}}^\mathrm{p}\left(\theta\right)}^H{ \dot{\mathbf{a}}^\mathrm{p}\left(\theta\right)}-\\\label{eq_CRB_conclusion2}&\frac{\|\dot{\mathbf{b}}^T{\mathbf{b}}^*{\mathbf{a}^\mathrm{p}\left(\theta\right)}^H{\mathbf{a}^\mathrm{p}\left(\theta\right)}\|^2+\|{\mathbf{b}}^T{\mathbf{b}}^*{ \dot{\mathbf{a}}^\mathrm{p}\left(\theta\right)}^H{\mathbf{a}^\mathrm{p}\left(\theta\right)}\|^2}{\mathbf{b}^T{\mathbf{b}}^*{\mathbf{a}^\mathrm{p}\left(\theta\right)}^H{\mathbf{a}^\mathrm{p}\left(\theta\right)}}\Bigg]^{-1}\\\nonumber
   =&\frac{\sigma^2}{2|\alpha|^2}\left[\dot{\mathbf{b}}^T\dot{\mathbf{b}}^*{\mathbf{a}^\mathrm{p}\left(\theta\right)}^H{\mathbf{a}^\mathrm{p}\left(\theta\right)}\left(1-\frac{\|\dot{\mathbf{b}}^T\mathbf{b}^*\|}{\dot{\mathbf{b}}^T\dot{\mathbf{b}}^*{\mathbf{b}}^T\mathbf{b}^*}\right)\right.\\\label{eq_CRB_conclusion3}&\left.{\mathbf{b}}^T{\mathbf{b}}^*{ \dot{\mathbf{a}}^\mathrm{p}\left(\theta\right)}^H{ \dot{\mathbf{a}}^\mathrm{p}\left(\theta\right)}\left(1-\frac{\|{ \dot{\mathbf{a}}^\mathrm{p}\left(\theta\right)}^H{\mathbf{a}^\mathrm{p}\left(\theta\right)}\|}{{ \dot{\mathbf{a}}^\mathrm{p}\left(\theta\right)}^H{ \dot{\mathbf{a}}^\mathrm{p}\left(\theta\right)}{\mathbf{a}^\mathrm{p}\left(\theta\right)}^H{\mathbf{a}^\mathrm{p}\left(\theta\right)}}\right)\right]^{-1}\\\nonumber
   =&\frac{\sigma^2}{2|\alpha|^2}\left[{\mathbf{b}}^\mathrm{p}\left(\theta \right)^T \mathbf{R}_X{\mathbf{b}}^\mathrm{p}\left(\theta \right)^*\dot{\mathbf{a}}^\mathrm{p}\left(\theta \right)^H\dot{\mathbf{a}}^\mathrm{p}\left(\theta \right)\Gamma_\mathrm{I}\left(\theta \right)+\right. \\\label{eq_CRB_conclusion1}&\left. \dot{\mathbf{b}}^\mathrm{p}\left(\theta \right)^T\mathbf{R}_X\dot{\mathbf{b}}^\mathrm{p}\left(\theta \right)^*{\mathbf{a}}^\mathrm{p}\left(\theta \right)^H{\mathbf{a}}^\mathrm{p}\left(\theta \right)\Gamma_\mathrm{S}\left(\theta \right) \right]^{-1},\end{align}
with
\begin{align}\nonumber\label{eq_CRB_Gamma_APP}
    &\Gamma_\mathrm{I/S}(\theta) \overset{M_\mathrm{I/S}\uparrow}{=}1-\\
    &\frac{\left[\sum_{l}F_{\mathrm{I/S}}^2(\theta,~l)\sin\left(\theta-\phi_l \right)\right]^2}{\sum_{l}F_{\mathrm{I/S}}^2(\theta,~l)\sum_{l}F_{\mathrm{I/S}}^2(\theta,~l)\left[1+\cos^2\left(\theta-\phi_l \right)\tan^2(\pi/L)/3\right]},
\end{align}
\end{subequations}
 {where ${\mathbf{a}}^\mathrm{p}\left(\theta \right)\triangleq \mathbf{F}_{\mathrm{S}}\left(\theta \right){\mathbf{a}}\left(\theta \right),~
{\mathbf{b}}^\mathrm{p}\left(\theta \right)\triangleq\mathbf{F}_{\mathrm{I}}\left(\theta \right){\mathbf{b}}\left(\theta \right),~
\dot{\mathbf{a}}^\mathrm{p}\left(\theta \right)\triangleq\frac{\partial {\mathbf{a}}^\mathrm{p}\left(\theta \right)}{\partial\theta},~
\dot{\mathbf{b}}^\mathrm{p}\left(\theta \right)\triangleq\frac{\partial {\mathbf{b}}^\mathrm{p}\left(\theta \right)}{\partial\theta}$, and $
\mathbf{R}_X\triangleq\mathbf{X}^*\mathbf{X}^T$. Additionally, in \eqref{eq_CRB_conclusion2} and \eqref{eq_CRB_conclusion3}, we use the abbreviations $\mathbf{b}\triangleq\mathbf{X}^H\mathbf{b}^\mathrm{p}(\theta)$ and $\dot{\mathbf{b}}\triangleq\mathbf{X}^H\dot{\mathbf{b}}^\mathrm{p}(\theta)$ for brevity.} In \eqref{eq_CRB_Gamma_APP}, $\Gamma_\mathrm{I/S}(\theta)$ is only dependent on the antenna patterns, but regardless of $N_{\mathrm{I/S}}$ or $M_{\mathrm{I/S}}$. ${\mathbf{a}}^\mathrm{p}\left(\theta \right)$ and  ${\mathbf{b}}^\mathrm{p}\left(\theta \right)$ are the receive steering vector embedded with the gain of the antennas at sensors and the transmit steering vector embedded with the gain of the antennas at
IS elements, respectively. The proof of \eqref{eq_CRB_Gamma_APP} is shown in  Appendix \ref{appendix_CRB_largeM_app}.

\begin{prop}\label{prop1}
{When the IS elements and the sensors share the same architecture (e.g., $N_\mathrm{I}=N_\mathrm{S}\triangleq N$ and $F_\mathrm{I}\left(\theta \right)=F_\mathrm{S}\left(\theta \right)\triangleq F\left(\theta \right)$, hence $M_\mathrm{I}=M_\mathrm{S}\triangleq M$,  $\mathbf{F}_{\mathrm{I}}\left(\theta \right)=\mathbf{F}_{\mathrm{S}}\left(\theta \right)\triangleq \mathbf{F}\left(\theta \right)$ and $\Gamma_{\mathrm{I}}\left(\theta \right)=\Gamma_{\mathrm{S}}\left(\theta \right)\triangleq \Gamma\left(\theta \right)$), the CRB in \eqref{eq_CRB_conclusion} can be further simplified as following, which is inversely proportional to two physical properties of the multi-sector IS. }
\begin{subequations}\label{eq_CRB_appr}
   \begin{align}
     \mathrm{CRB}\left(\theta \right)\overset{M\uparrow}{=}& \frac{\sigma^2}{4|\alpha|^2\Gamma\left(\theta \right)}\left[e \left(\theta \right) r^2 \left(\theta \right) \right]^{-1},\\
    \label{eq_power} \mbox{with}~
    e\left(\theta \right)=&\|{\mathbf{b}}^\mathrm{p}\left(\theta \right)^{H}\mathbf{X}^{H}\|^2 ,
    \\   \label{eq_rate_angle}
     r\left(\theta \right)=&\|\dot{\mathbf{a}}^\mathrm{p}\left(\theta \right)\|,
    \end{align} 
\end{subequations}
where $e\left(\theta \right)$ is the probing power on target and $r\left(\theta \right)$ is the rate of target angle.
\end{prop}
\begin{proof}
Please refer to Appendix B.
\end{proof}

\begin{nrem}\label{rem:2}{Both $e\left(\theta \right)$ in \eqref{eq_power} and $r^2\left(\theta \right)$ in \eqref{eq_rate_angle} have intuitive interpretations. The probing power on the target, $e\left(\theta \right)$, affects sensing performance since more power to be reflected by the target indicates a higher signal-to-noise ratio (SNR) for estimating $\theta$. Besides, $e\left(\theta \right)$ in \eqref{eq_power} is closely related to $\mathbf{X}$ and $F(\theta)$. On the other hand, the rate of target angle, $r\left(\theta \right)$, describes the sensitivity of the sensors against $\theta$, adhering to the sensor's geometry $\mathbf{P}_{\mathrm{S},~l}$ and $F(\theta)$, as will be shown in Section \ref{sec3-2:e_r2}. The higher the rate of target angle, the better we can distinguish between close target angles while being less affected by the noise \footnote{The rate of target angle determines the sharpness of the peak of the ambiguity function, which ideally is a Dirac function for radar sensing.}.} \end{nrem}

\begin{nrem}\label{rem:Gamma}
$\Gamma(\theta)$ in \eqref{eq_CRB_Gamma_APP} is independent of $N_{\mathrm{I/S}}$, $M_{\mathrm{I/S}}$, $\mathbf{X}$, and $P^{\mathrm{tr}}$. Moreover, for different configurations, $\Gamma(\theta)$ can also be shown to be strictly upper-bounded by $1$ with slight fluctuations across $\theta$ (as explained in Appendix \ref{app:Gamma}). Hence,  $\Gamma(\theta)$ will not be the focus of our analysis in the following discussion. 
\end{nrem}

In the following sub-section, we assume that the IS elements and sensors are symmetric, as assumed in Proposition \ref{prop1}. This assumption, for balancing generality and tractability, not only simplifies the analysis but also facilitates a clear understanding of the core components that directly influence the CRB. However, it is important to note that our ML estimator in   \eqref{eq_MLE}  and the CRB in   \eqref{eq_CRB_conclusion} can be adapted to arbitrary architectures and antenna patterns.

\subsection{General insights into  $e\left( \theta\right)$ and $r^2\left( \theta\right)$ as fundamental CRB components}\label{sec3-2:e_r2}
This sub-section provides analytical insights into the two components, $e\left( \theta\right)$ and $r^2\left( \theta\right)$, with arbitrary antenna patterns. First, $e\left( \theta\right)$ in \eqref{eq_power} can be further expressed as
\begin{subequations} \label{eq_power_insights}
    \begin{align}\label{eq_power_insights1}
      e\left(\theta \right) =&\mathbf{b}^\mathrm{p}\left(\theta\right)^{T}\mathbf{X}\left[\mathbf{b}^\mathrm{p}\left(\theta\right)^{T}\mathbf{X}\right]^{H}\\\label{eq_power_insights2}
      =&{P^{\mathrm{tr}}}\|\sum_l F\left(\theta,~l\right)\mathbf{b}_l\left(\theta\right) \|^2\\\label{eq_power_insights3}
      \overset{M \uparrow}{=}&{P^{\mathrm{tr}}}\sum_l {F\left(\theta,~l\right)}^2\|\mathbf{b}_l\left(\theta\right) \|^2\\
   =&P^{\mathrm{tr}}N\sum_{l}{F\left(\theta,~l\right)}^2/L,
\end{align}
\end{subequations}
where   \eqref{eq_power_insights1} comes from the definition of $e\left(\theta \right)$ in \eqref{eq_power},   \eqref{eq_power_insights2} comes from $\mathbf{X}$ defined in \eqref{eq_tran_signal2}, and   \eqref{eq_power_insights3} comes from $\mathbf{b}_{l_1}\left(\theta\right)^H\mathbf{b}_{l_2}\left(\theta\right)\overset{M\uparrow}{=}0$ for $l_1\neq l_2$, as demonstrated in \eqref{eq_APP_diff_sector_0} of Appendix \ref{appendix_CRB_app}.

Second, $r\left( \theta\right)$ in \eqref{eq_rate_angle} can be expressed as
\begin{subequations}\label{eq_rate_insight}
    \begin{align}
   \nonumber 
    &r^2\left(\theta \right)=\dot{\mathbf{a}}^{p}\left(\theta \right)^T\dot{\mathbf{a}}^{p}\left(\theta \right)^*\\
    =&\sum_l \left[\frac{\partial {F}\left( \theta,~l\right)\mathbf{a}_l\left( \theta\right)}{\partial \theta}\right]^H\frac{\partial {F}\left( \theta,~l\right)\mathbf{a}_l\left( \theta\right)}{\partial \theta}\\\nonumber {=}&
\sum_{l}\pi^2{F\left(\theta,~l \right) }^2\left( -\frac{N^3-NL^2}{12L^3}+\frac{N^3}{4L^3\tan^2\left( \pi/L\right)}\right)\sin^2\left(\theta-\phi_l \right)\\\label{eq_rate_insight2}&+\sum_{l}\pi^2{F\left(\theta,~l \right) }^2\frac{N^3-NL^2}{12L^3}+{\dot{F}\left(\theta,~l \right) }^2\frac{N}{L},
\end{align}
\end{subequations}
where   \eqref{eq_rate_insight2} comes from \eqref{eq_app_a_derivative_a_derivative} in Appendix \ref{appendix_CRB_largeM_app}.

Next, we provide a detailed analysis of \eqref{eq_power_insights} and \eqref{eq_rate_insight} with the specific configurations of $L=2,~3,~4$.

\subsubsection{For $L=2$} the coverage 
of each sector is orthogonal, and hence the target is sensed by only one sector. We assume the target is illuminated by sector $2$ as depicted in Fig. \ref{fig_system_model}. From \eqref{eq_power_insights} and \eqref{eq_rate_insight},  $e\left(\theta \right)$ and $ r^2\left(\theta \right)$ are rewritten as
\begin{subequations}\label{eq_L_2}
  \begin{align}\label{eq_L_2_e}
    e\left(\theta \right)|_{L=2}=&\frac{P^{\mathrm{tr}}F^2(\theta,~2)}{2}N, \\\label{eq_L_2_r}
    r^2\left(\theta \right)|_{L=2}\overset{N\uparrow}{\approx}&\frac{\pi^2F^2(\theta,~2)\cos^2\left(\theta\right)}{96}N^3.
\end{align}  
\end{subequations}
{From \eqref{eq_L_2_r}, we can deduce that  $r^2(\theta)$ equals to $0$ for $\theta=\pi/2$ or for $\theta=-\pi/2$ because of the term $\cos^2\left(\theta\right)$. This results in infinite CRB in \eqref{eq_CRB_appr}, i.e., the geometry of $L=2$ cannot provide accurate estimation when the target is around $\theta=\pi/2$ or $\theta=-\pi/2$.  This observation suggests that the $2$-sector IS geometry, which coincides with the conventional STARS, has a blind sensing area, and hence might not be suitable for the sensing scenarios without any prior knowledge of the target angle. }

\subsubsection{For $L=3$} we consider a periodicity of $\theta\in \left(0, ~2\pi/3\right)$, which is divided into two phases, i.e., $\theta\in \left(0, ~\pi/3\right)$ where  the target is only illuminated by sector $3$ and $\theta\in \left( \pi/3, ~2\pi/3\right)$ where  the target is illuminated by sector $1$ and sector $3$ simultaneously.

For $\theta\in \left(0, ~\pi/3\right)$, $e\left(\theta \right)$ and $ r^2\left(\theta \right)$ are re-expressed as
\begin{subequations}\label{eq_L_31}
  \begin{align}\label{eq_L_31_e}
    e\left(\theta \right)|_{L=3}=&\frac{P^{\mathrm{tr}}F^2(\theta,~3)}{3}N, \\\label{eq_L_31_r}
    r^2\left(\theta \right)|_{L=3}\overset{N\uparrow}{\approx}&\frac{\pi^2F^2\left(\theta,~3 \right)}{324}N^3,
\end{align}  
\end{subequations}
while for $\theta\in \left( \pi/3, ~2\pi/3\right)$, $e\left(\theta \right)$ and $ r^2\left(\theta \right)$ are re-expressed as
\begin{subequations}\label{eq_L_32}
  \begin{align}\label{eq_L_32_e}
    e\left(\theta \right)|_{L=3}=&\frac{P^{\mathrm{tr}}\sum_{l=1,3}F^2(\theta,~l)}{3}N, \\\label{eq_L_32_r}
    r^2\left(\theta \right)|_{L=3}\overset{N\uparrow}{\approx}&\frac{\sum_{l=1,3}\pi^2{F^2\left(\theta,~l \right) }}{324}{N^3}.
\end{align}  
\end{subequations}
 {From \eqref{eq_L_31} and   \eqref{eq_L_32}, we can deduce that, if $F(\theta)$ is uniform with respect to (w.r.t.) $\theta$ (corresponding to the half-space isotropic antenna patterns as will be described later in Section \ref{sec3-3: scaling_summary}), $e(\theta)$ and $r^2(\theta)$ in the second phase, i.e., $\theta\in \left( \pi/3, ~2\pi/3\right)$, are generally larger than their counterparts in the first phase, i.e., $\theta\in \left(0, ~\pi/3\right)$, as the target in the second phase is illuminated by two sectors, compared with only one sector in the first phase. This will result in non-uniform MSE performance across $\theta$, which is undesirable for full-space wireless sensing where the performance of the worst case of angle estimation should be guaranteed.}

\subsubsection{For $L=4$} each target is illuminated by two aligned sectors. We assume the target is illuminated by sectors $1$ and $4$ as depicted in Fig. \ref{fig_system_model}. Hence, $e\left(\theta \right)$ and $ r^2\left(\theta \right)$ are rewritten as
\begin{subequations}\label{eq_L_4}
  \begin{align}\label{eq_L_4_e}
    e\left(\theta \right)|_{L=4}=&\frac{\sum_{l=1,~4} P^{\mathrm{tr}}F^2(\theta,~l)}{4}N, \\\label{eq_L_4_r}
    r^2\left(\theta \right)|_{L=4}=&\frac{\sum_{l=1,~4}F^2(\theta,~l)\left(2\sin^2(\theta-\phi_l)+1\right)\pi^2}{768}N^2.
\end{align}  
\end{subequations}
  {For $L=4$, if assuming uniform $F(\theta)$, we can readily deduce that $e\left(\theta \right)$ and $ r^2\left(\theta \right)$ become uniform across $\theta$. Hence, the MSE performance for the $4$-sector IS self-sensing is uniform w.r.t. $\theta$, which is desirable for full-space wireless sensing.} 

 {The expressions from \eqref{eq_L_2} to \eqref{eq_L_4} also indicate that
$F\left(\theta\right)$
directly affects $e\left(\theta \right)$ and $r^2\left(\theta \right)$, which further impacts the sensing performance. To better characterize this effect, in the following, we consider two specific antenna patterns, i.e., the half-space isotropic and half-space directive antenna patterns, to facilitate deriving their performance scaling laws for comparison. }

\subsection{Numerical scaling laws of $e\left( \theta\right)$ and $r^2\left( \theta\right)$ with specific antenna patterns} \label{sec3-3: scaling_summary}
To derive the numerical scaling laws,  {we first specify the mathematical expressions of the gain of antenna patterns in Section \ref{sec_AP}. Next, we substitute them into the equations derived in Section \ref{sec3-2:e_r2}, and obtain the numerical scaling laws in Section \ref{sec_scaling_table}.} 

\subsubsection{Gain of different antenna patterns} \label{sec_AP}
 {The gain of the half-space isotropic and the half-space directive antenna patterns are specifically given as follows.}

First, for the half-space isotropic antenna pattern,  {the antenna gain towards a target at $\theta_l$ in the L-CCS $x_l-y_l  $ is given by }
\begin{align}
\label{eq_antenna_pattern_iso_theta}  
G^{\mathrm{Iso}}\left(\theta_l\right)=&   \begin{cases}
2,&\cos(\theta_l) \geq 0,\\
0,&\text{otherwise,}
\end{cases}
\end{align}
for which the details are given in Appendix \ref{app:AP derive}.  {The gain of the half-space isotropic antenna pattern} has been normalized by the total radiated power for comparison fairness.
 
 {Second, for the half-space directive antenna pattern in \cite{li2023beyond}, the antenna gain towards a target at $\theta_l$ in the L-CCS $x_l-y_l  $ is given by (normalized by the total radiation power)}
\begin{align}\label{eq_antenna_pattern_dir}  
G^{\mathrm{Dir}}\left(\theta_l\right)=&  \begin{cases}
2\left(\alpha_{\mathrm{L}}+1\right)\cos^{\alpha_\mathrm{L}}\left(\theta_l \right),&\cos(\theta_l)\geq 0,\\
0,&\text{otherwise,}
\end{cases}\end{align}
with $\alpha_{\mathrm{L}}=\log \left(0.5\right)/\log \left(\cos\left(\frac{\pi}{L}\right)\right)$. $\alpha_{\mathrm{L}}$ is set to align the half-power beamwidth of the antenna's radiation pattern with the concentrated coverage of each sector, i.e., $2\pi/L$. Notice that for $L=2$, each sector has to cover $180^o$ space, and the half-space directive antenna pattern in \eqref{eq_antenna_pattern_dir} boils down to the half-space isotropic antenna pattern in   \eqref{eq_antenna_pattern_iso_theta}.  More details are given in Appendix \ref{app:AP derive}.

\subsubsection{Summary and analysis of the numerical scaling laws}\label{sec_scaling_table}
\begin{table*}
     \centering
     \begin{tabular}{ |p{0.3cm}|p{2.3cm}|p{4.8cm}|p{4.4cm}|p{1.4cm}|p{1.7cm}|  }
    \hline
    \multicolumn{6}{|c|}{The law of scaling terms (\textbf{$ \mathrm{CRB}\left( \theta\right)\approx\frac{\sigma^2}{4|\alpha|^2\Gamma(\theta)}\left[  e \left(\theta \right) r^2 \left(\theta \right) \right]^{-1}$})} \\
    \hline
    \textbf{$L$}   &  Antenna pattern &\textbf{ $ r^2\left( \theta\right)$ } &   \textbf{$ e\left( \theta\right)$} & \textbf{$\mathbb{E}_{\theta}\left\lbrace r^2\left(\theta \right) \right\rbrace$  } &  \textbf{$\mathbb{E}_{\theta}\left\lbrace e\left(\theta \right) \right\rbrace$  }    \\
    \hline
    \multirow{2}{*}{2}& Isotropic&  \multirow{2}{*}{$\pi^2{N^3}\cos^2\left(\theta\right)/48$}&\multirow{2}{*}{$P^{\mathrm{tr}}{N^2}$ }&\multirow{2}{*}{$0.102  N^{3}$ }  & \multirow{2}{*}{$P^{\mathrm{tr}}{N^2}$  }  \\
     \cline{2-2}
      & Directive&  & & &   \\
      \hline
     \multirow{4}{*}{3}& Isotropic, $\theta\in \left(0, ~\pi/3\right)$ &${\pi^2N^3}/{162}$& ${2P^{\mathrm{tr}}N^2}/{3}$ & $0.061N^3$& ${0.67P^{\mathrm{tr}}N^2}$ \\
    \cline{2-6}
     & Isotropic, 
     $\theta\in \left(\pi/3, ~2\pi/3\right)$ &${\pi^2N^3}/{81}$& ${4P^{\mathrm{tr}}N^2}/{3}$ & $ 0.121N^3$& ${1.33P^{\mathrm{tr}}N^2}$ \\\cline{2-6} &Directive, $\theta\in \left(0, ~\pi/3\right)$
&${\pi^2N^3\cos\left({\theta}-\phi_3 \right)}/{81}$& ${4P^{\mathrm{tr}}N^2\cos\left(\theta-\phi_3\right)}/{3}$ & \multirow{2}{*}{$0.116N^3$} & \multirow{2}{*}{${1.33P^{\mathrm{tr}}N^2}$}\\
     \cline{2-4}
     & Directive, 
     $\theta\in \left(\pi/3, ~2\pi/3\right)$ &${\pi^2N^3\sin\left(\theta \right)}/{81}$& ${4P^{\mathrm{tr}}N^2\sin(\theta)}/{3}$ &  &  \\
     \hline
     \multirow{2}{*}{4}& Isotropic &$\pi^2{N^3}/{96}$& $P^{\mathrm{tr}}{N^2}$ & $ 0.102N^3$& $P^{\mathrm{tr}}{N^2}$ \\
    \cline{2-6}
    & Directive &${\left[3-\cos\left(4\theta\right)\right]\pi^2 N^3}/{256 }$& ${3P^{\mathrm{tr}}N^2}/{2}$&$ 0.116 N^{3}$& ${1.5P^{\mathrm{tr}}N^2}$\\
    \hline
\end{tabular}
     \caption{Table of the scaling terms of $e\left(\theta \right)$ and $ r^2\left(\theta \right)$ by employing the antenna patterns of \eqref{eq_antenna_pattern_iso_theta} and  \eqref{eq_antenna_pattern_dir} (\textit{Notice that only the highest order term of the scaling terms in Section \ref{sec3-3: scaling_summary} is adopted.}).  For details, please see Remarks \ref{rem:2} and \ref{rem:CRB}.  For the derivation details, please see Appendix \ref{Appendix_e_r}.}
     \label{tab2}
 \end{table*}

Next, we substitute the antenna patterns of \eqref{eq_antenna_pattern_iso_theta} and  \eqref{eq_antenna_pattern_dir} into the derived $e\left(\theta \right)$ and $ r^2\left(\theta \right)$, and obtain their numerical scaling laws which have more insightful forms. Additionally, we also take expectations of these metrics over $\theta$ for an overall performance comparison. The results are summarized in Table \ref{tab2}, which provides a clear overview of the impacts of multi-sector geometries and antenna patterns on the sensing performance. In the following, ''isotropic antenna pattern'' is adopted as an abbreviation for ''half-space isotropic antenna pattern'' and ''directive antenna pattern'' is adopted as an abbreviation for ''half-space directive antenna pattern''.

\begin{nrem}\label{rem:CRB}
From Table \ref{tab2}, we observe that:
    \begin{itemize}
        \item {Using the directive antenna pattern always achieves better performance than using the isotropic antenna pattern,  $\forall~L$, by comparing \footnote{For $L=3$, 
        averaging $ e\left(\theta \right)$ and $r^2\left(\theta \right)$  is taken over all angles, i.e., 
        $0\leq\theta\leq \pi/3$ and $ \pi/3\leq\theta\leq 2\pi/3$.} $\mathbb{E}\left\lbrace e\left(\theta \right) \right\rbrace$ and $\mathbb{E}\left\lbrace r^2\left(\theta \right) \right\rbrace$ straightforwardly, where the directive counterparts are generally larger than the isotropic counterparts except the case of $L=2$. }
        \item {For the isotropic antenna pattern,  only the $4$-sector IS configuration possesses (almost) uniform $ e\left(\theta \right)$ and $r^2\left(\theta \right)$ across $\theta$, and will be shown to perform the best in terms of MSE by simulations. In contrast, the geometry of $L=2$ shows dynamic $r^2(\theta)$ because of the term $\cos^2\left(\theta\right)$.  The geometry of $L=3$ also features highly varying  $ e\left(\theta \right)$ and $r^2\left(\theta \right)$   w.r.t. $\theta$. }
        \item {For the directive antenna pattern, $\mathbb{E}\left\lbrace e\left(\theta \right) \right\rbrace$ or $\mathbb{E}\left\lbrace r^2\left(\theta \right) \right\rbrace$ increases as a function of $L$. This naturally makes $L=4$ (among $L=2,~3,~4$ configurations) the best geometry. Noticeably, given the directive antenna pattern, $ e\left(\theta \right)$ and $r^2\left(\theta \right)$ for $L=3$ between different phases (i.e., $0\leq\theta\leq \pi/3$ and $\pi/3\leq\theta\leq 2\pi/3$) share similar mathematical forms, which makes $ e\left(\theta \right)$ and $r^2\left(\theta \right)$ roughly uniform across $\theta$ (more details given in Appendix \ref{Appendix_e_r}). As a result, with directive antenna patterns, the $3$-sector IS significantly outperforms the $2$-sector IS, as will be demonstrated by the simulation results. } 
       
    \end{itemize}

\end{nrem}

\section{Simulation Results}
\label{sec4: simulation}
 {In this section, we evaluate the performance of the proposed multi-sector IS self-sensing system via numerical results.  The simulations are performed on an 802.11p standard (wireless access
in vehicular environments), with sub-carrier frequency $5.19$ GHz, noise power $\sigma^2=-80$ dBm, and LoS channels. We assume the target scattering coefficient $\alpha_\mathrm{T}=0$ dB, and assume the target is uniformly randomly distributed on a circle which centers at the $x_0-y_0$'s origin with a radius of $519$ m. In terms of the configuration of the multi-sector IS self-sensing structure, we assume the distance between the active source controller and sector $1$ is $d_{\mathrm{CI}}=0.5$ m to guarantee the far-field assumption. Particularly, we mainly focus on evaluating and analyzing the sensing performance of multi-sector IS geometries for $L=2,~3,~4$ as shown in Fig. \ref{fig_system_model}, where $L=2$ corresponds to the conventional STARS configuration and serves as a benchmark for $L>2$. In addition, we also provide simulation results for general configurations, for example, $L=5$ and $6$. For simplicity, we consider ULA for IS elements and sensors at each sector as shown in Fig. \ref{fig_system_model}, and we set $N_\mathrm{S}=N_\mathrm{I}\triangleq N$ and $\mathbf{F}_{\mathrm{I}}\left(\theta\right)=\mathbf{F}_{\mathrm{S}}\left(\theta\right)\triangleq\mathbf{F}\left(\theta\right)$. In the following, the examined MSE is divided into two parts, the instantaneous MSE   w.r.t. target angles, and the overall MSE. Specifically, the instantaneous MSE is defined as $\mathrm{MSE}\left(\theta\right)\triangleq \mathbb{E}_{z}\left\lbrace \|\widehat{\theta}-\theta\|^2\right\rbrace$ where $\widehat{\theta}$ is the ML estimation of $\theta$ and $z$ is the random AWGN, and the overall MSE is defined as $\mathrm{MSE}\triangleq \mathbb{E}_{\theta}\left\lbrace \mathrm{MSE}\left(\theta\right)\right\rbrace$. Similarly, the overall $\mathrm{CRB}$ is defined as 
$\mathrm{CRB}\triangleq \mathbb{E}_{\theta}\left\lbrace \mathrm{CRB}\left(\theta\right)\right\rbrace$.}

\subsection{The instantaneous MSE performance}\label{sec4-1:sim1}
\begin{figure}
    \centering
    \includegraphics[width=3 in]{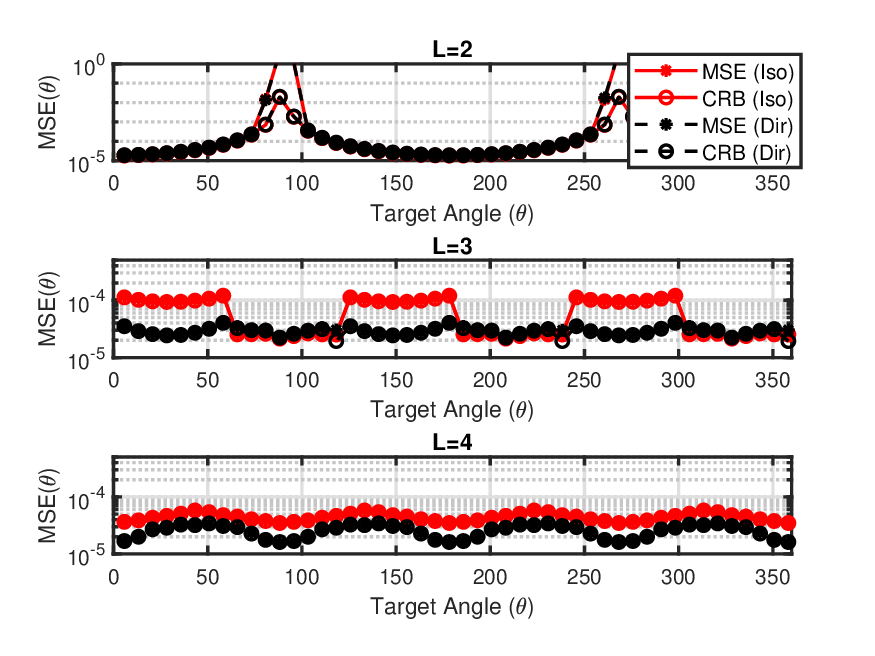}
    \caption{The instantaneous $\mathrm{MSE}\left(\theta\right)$ and $\mathrm{CRB}\left(\theta\right)$ w.r.t. $\theta$ for $N=24$,  $P^{\mathrm{tr}}=45$ dBm and both half-space isotropic and half-space directive antenna patterns. Herein, ''Iso'' refers to employing the half-space isotropic antenna pattern and ''Dir'' refers to employing the half-space directive antenna pattern. }
    \label{fig_sim1_mse}
\end{figure}

\begin{figure*}[!t]
\centering
\subfloat[]{\includegraphics[width=3.0 in]{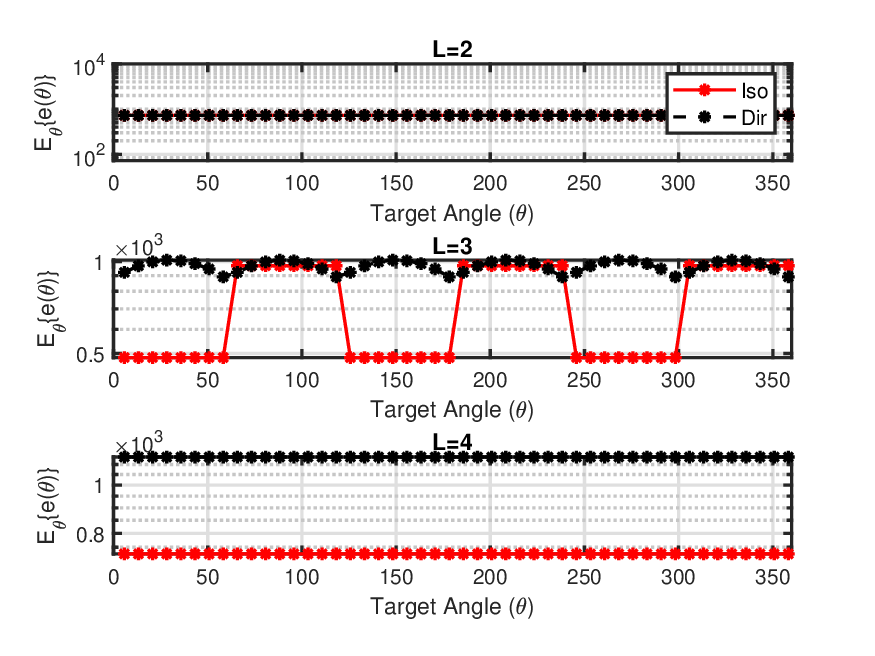}%
\label{fig_pow}}
\hfil
\subfloat[]{\includegraphics[width=3.0 in]{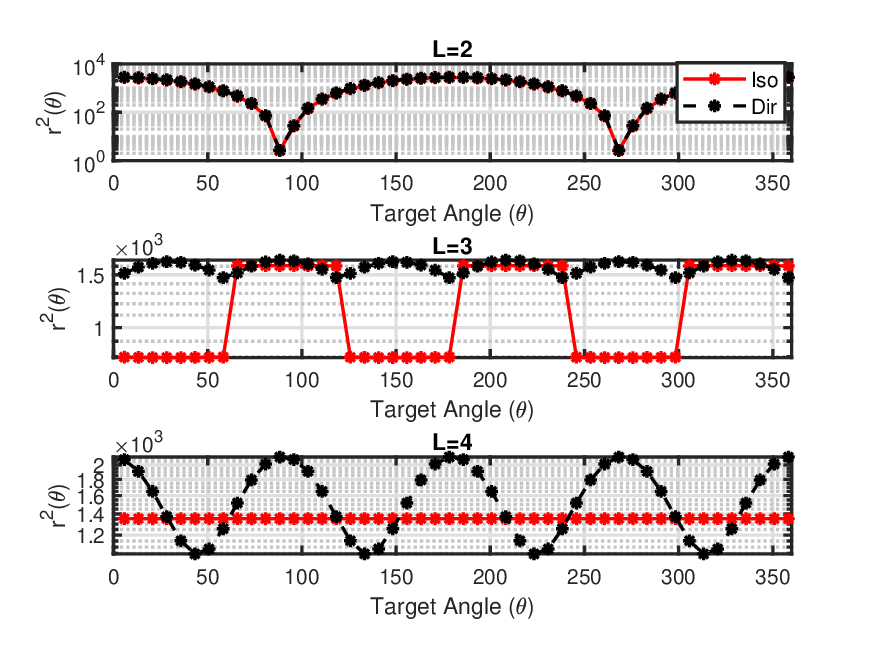}%
\label{fig_rate_of_angle}}
\caption{Corresponding to Fig. \ref{fig_sim1_mse}, the two fundamental components of the CRB therein, with (a) the probing power on the target across $\theta$, (b) the squared rate of target angle across $\theta$. The numerical results are in accordance with Table \ref{tab2}.}
\label{fig_deter_L32}
\end{figure*}

{We first evaluate the instantaneous $\mathrm{MSE}\left( \theta\right)$ of the proposed multi-sector IS self-sensing system as shown in Fig. \ref{fig_sim1_mse}, where different multi-sector IS geometries ($L=2,~3,~4$) are compared under different antenna patterns.   In Fig. \ref{fig_sim1_mse}, we also plot the corresponding $\mathrm{CRB}\left( \theta\right)$ in \eqref{eq_CRB_conclusion}, which is shown to be precisely aligned with $\mathrm{MSE}\left( \theta\right)$. From Fig. \ref{fig_sim1_mse}, we observe that $\mathrm{MSE}\left( \theta\right)$ for $L=4$ shows the best stability across $\theta$.  For the isotropic antenna pattern, there is a notable performance degradation when $\theta$ approaches $90^o$ or $270^o$ for $L=2$, whereas, for $L=3$, performance significantly worsens when the target is only illuminated by one sector (i.e., $0^o\leq \theta \leq 60^o$). In addition, using the directive antenna pattern generally achieves better $\mathrm{MSE}\left( \theta\right)$ than using the isotropic antenna pattern, except the case of $L=2$ where the directive antenna pattern in \eqref{eq_antenna_pattern_dir} boils down to the isotropic antenna pattern in \eqref{eq_antenna_pattern_iso_theta}.}
 
 {To provide more insights into Proposition \ref{prop1},  Fig. \ref{fig_pow} and Fig. \ref{fig_rate_of_angle} depict the two fundamental components of  the $\mathrm{CRB}\left(\theta \right)$ in Fig. \ref{fig_sim1_mse}, i.e., $e\left(\theta \right)$ and $ r^2\left(\theta \right)$.  For $L=2$, it is observed that $e\left(\theta \right)$ is stable across different $\theta$ while  $r^2\left(\theta \right)$ deteriorates sharply when $\theta$ approaches around $90^o$ or $270^o$. The degradation of $r^2\left(\theta \right)$ results in more fluctuating $\mathrm{CRB}\left(\theta \right)$ and  $\mathrm{MSE}\left(\theta \right)$ in  Fig. \ref{fig_sim1_mse}, which highlights the essential role played by the geometry of the sensors. In contrast, for $L=3$, an interesting observation is that, both $e\left(\theta \right)$ and  $r^2\left(\theta \right)$ exhibit more stability with higher values across $\theta$ when using the directive antenna pattern, compared with using the isotropic antenna pattern. In this context, the adopted directive antenna pattern compensates for the performance deficiency of the $3$-sector IS geometry, i.e., targets at different angles are illuminated by a different number of sectors. The compensation is achieved by concentrating each sector's radiation more towards the angles where the target is illuminated by one sector only. In the contrary, for targets that can be seen from two sectors, the radiation signals from both sectors become weaker when employing the directive antenna pattern compared with the isotropic antenna pattern. Finally, for  $L=4$, we observe that the directive antenna pattern offers significant gain on $e(\theta)$, despite leading to slight fluctuations on $r^2(\theta)$. This results in better $\mathrm{MSE}\left(\theta \right)$ in  Fig. \ref{fig_sim1_mse}, which highlights the advantages of the directive antenna pattern for enhancing the power efficiency.}

\subsection{The {overall  MSE} performance} \label{sec4-1:sim2}

\begin{figure}
    \centering
    \includegraphics[width=3 in]{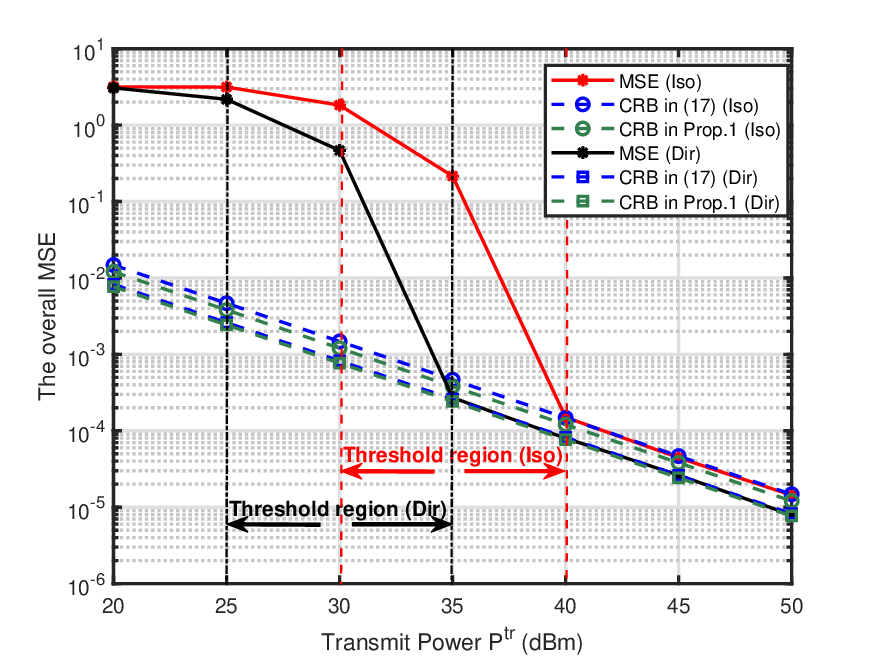}
    \caption{The overall $\mathrm{MSE}$ averaging over $\theta$, as a function of transmit power (in proportional to SNR), with different antenna patterns for $N=24$ and $L=4$.}
    \label{fig_diff_SNR}
\end{figure}

  In this sub-section, we first plot the overall $\mathrm{MSE}$ and the overall $\mathrm{CRB}$ as a function of transmit power for $L=4$, as shown in Fig. \ref{fig_diff_SNR}. Specifically, Fig. \ref{fig_diff_SNR} compares the overall $\mathrm{MSE}$ with the overall $\mathrm{CRB}$ in \eqref{eq_CRB_conclusion} and the approximated $\mathrm{CRB}$ in \eqref{eq_CRB_appr} (in Proposition \ref{prop1} with $\Gamma(\theta)$ in \eqref{eq_CRB_appr} being substituted by $\Gamma(\theta)=1$ in accordance with Remark \ref{rem:Gamma}). It shows that the approximated $\mathrm{CRB}$ from \eqref{eq_CRB_appr} is closely aligned with the precise $\mathrm{CRB}$ in \eqref{eq_CRB_conclusion} and the $\mathrm{MSE}$ in the high SNR region. Moreover, from Fig. \ref{fig_diff_SNR}, it becomes more evident that using the directive antenna pattern achieves better MSE performance compared with using the isotropic antenna pattern. Specifically, the overall $\mathrm{MSE}$ of using the directive antenna pattern exhibits the threshold region phenomenon\footnote{The threshold region phenomenon denotes a sensing phenomenon of a drastic improvement in sensing performance from being poor to being excellent. The SNR region that exhibits this improvement is named the threshold region. Usually, the MSE performance closely aligns with CRB after the threshold region.} ahead of that of using the isotropic antenna pattern by roughly $5$ dBm transmit power.

\begin{figure}[t]
\centering
\includegraphics[width=3 in]{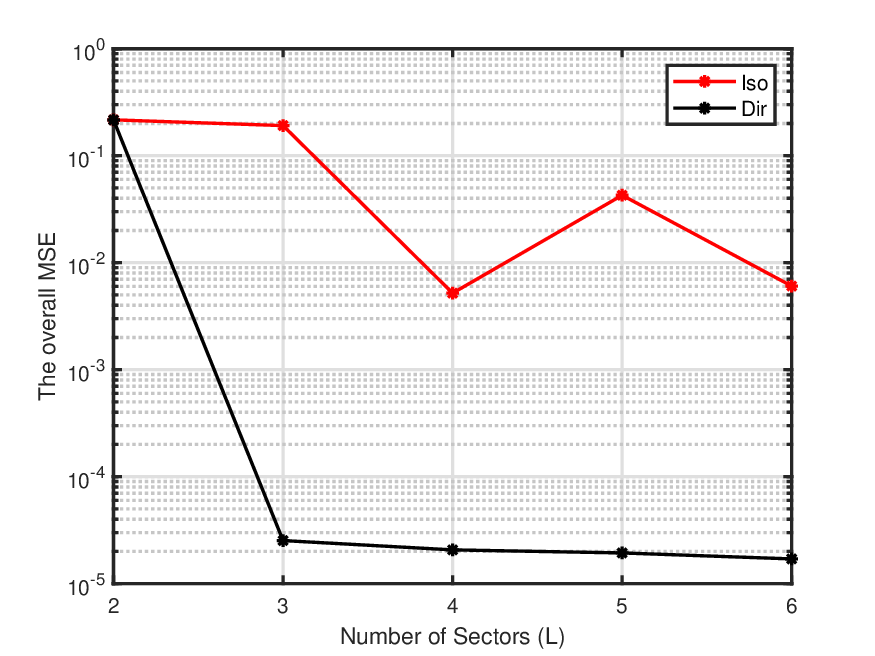}%
\caption{ The average MSE by taking the expectation of $\theta$, as a function of the number of sectors ($L$) for $N=60$ with different numbers of sectors, with $P^{\mathrm{tr}}=30$ dBm.}
\label{fig_diff_L}
\end{figure}

Next, Fig. \ref{fig_diff_L} compares the overall $\mathrm{MSE}$ between different multi-sector IS configurations with different $L$. Firstly, Fig. \ref{fig_diff_L} shows that for both antenna patterns, $L=4$ gives the best overall $\mathrm{MSE}$ among $L=2,~3,~4$, which coincides with the simulation results in Section \ref{sec4-1:sim1}. {Fig. \ref{fig_diff_L} also explores a broader range of $L$, i.e., including $L=5$ and $L=6$. It is observed that, with the isotropic antenna pattern, the overall $\mathrm{MSE}$ of odd $L$ is worse than that of even $L$ in the multi-sector IS self-sensing system\footnote{$L=2$ is excluded here since it represents the conventional STARS configuration and serves as a benchmark.}. This discrepancy mirrors the comparison between $L=3$ and $L=4$ as discussed in Section  \ref{sec4-1:sim1}, where the multi-sector IS geometry with an even $L$ benefits from relatively more uniform $\mathrm{MSE}\left(\theta \right)$.  Moreover, with the directive antenna pattern, the overall $\mathrm{MSE}$ consistently improves with the increasing $L$. This improvement stems from the higher directivity of directive antenna patterns with larger $L$, which particularly benefits $e(\theta)$ as explained in Section \ref{sec3-3: scaling_summary}.

Finally, Fig. \ref{fig_scaling} validates the numerical scaling laws in Section \ref{sec3-3: scaling_summary} as a function of the number of IS elements/sensors ($N$) given different $L$ and different antenna patterns. The figure shows that increasing $N$ enhances both $\mathbb{E}_{\theta}\left\lbrace e(\theta)\right\rbrace$ and $\mathbb{E}_{\theta}\left\lbrace r^2(\theta)\right\rbrace$ for the target angle estimation, which is in accordance with Table \ref{tab2}. Moreover, Fig. \ref{fig_scaling} illustrates the advantages of employing a directive antenna pattern for both 
$\mathbb{E}_{\theta}\left\lbrace e(\theta)\right\rbrace$ and $\mathbb{E}_{\theta}\left\lbrace r^2(\theta)\right\rbrace$, particularly in terms of $\mathbb{E}_{\theta}\left\lbrace e(\theta)\right\rbrace$.

\begin{figure}[h]
    \centering
    \begin{subfigure}{0.4\textwidth}
        \centering
        \includegraphics[width=3 in]{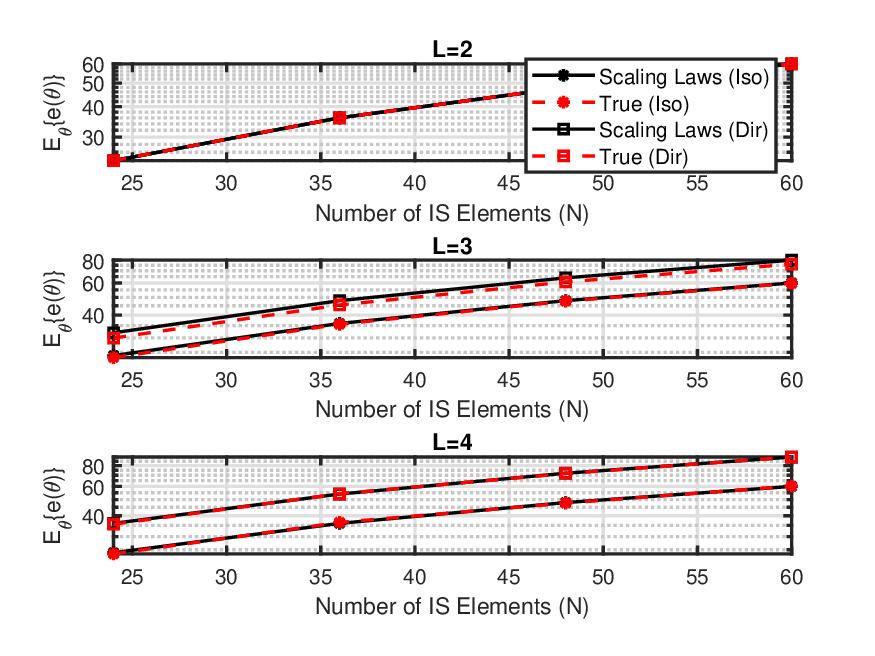}
        \caption{Probing power on the target}
        \label{fig_scaling_pow}
    \end{subfigure}%
    \hfill 
    \begin{subfigure}{0.4\textwidth}
        \centering
        \includegraphics[width=3 in]{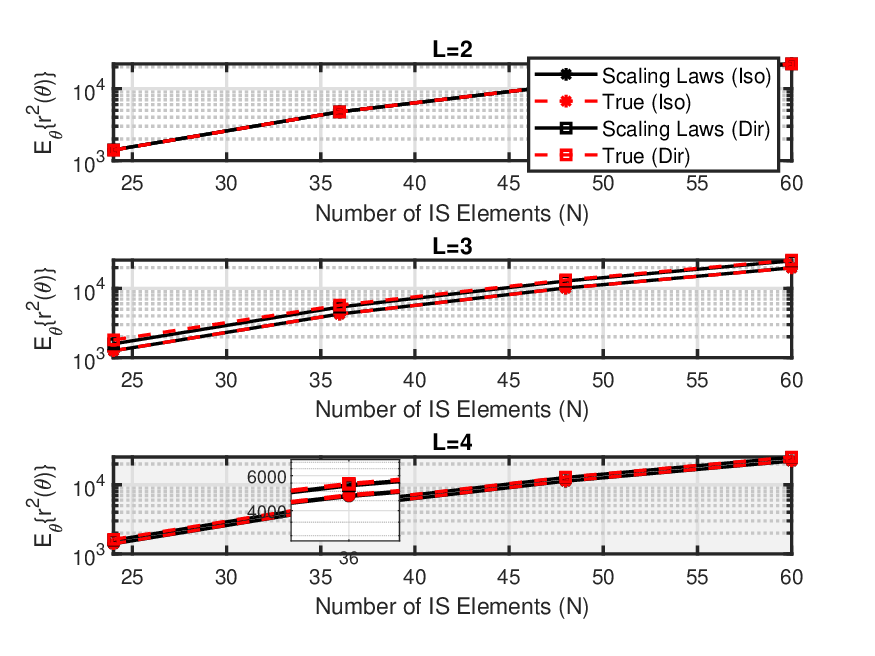}
        \label{fig_scaling_rate}
        \caption{Squared rate of target angle}
    \end{subfigure}
    \caption{The scaling laws in Table \ref{tab2} in Section \ref{sec3-3: scaling_summary} as a function of the number of IS/sensors elements, with (a) the scaling laws of the probing power on the target averaged over target angles; (b) the scaling laws of the squared rate of target angle averaged over target angles.}
    \label{fig_scaling}
\end{figure}


\section{Conclusions}\label{sec5:conclusion}
This paper proposed a new multi-sector IS self-sensing system to achieve full-space coverage for wireless sensing. Specifically, we developed an ML estimator of the target angle for the proposed multi-sector IS self-sensing system, along with the corresponding performance limits in terms of the CRB. The analysis of the CRB revealed that it primarily consists of two fundamental components, the probing power on the target and the squared rate of target angle. We showed that the multi-sector IS not only benefited from improved probing power but also enhanced the squared rate of target angle by offering more freedom in geometries. Moreover, it was verified that using directive antenna patterns can further enhance the sensing performance of multi-sector IS. The simulation results revealed that, among the geometries examined ($L=2,~3,~4$), the $4$-sector IS configuration achieves the best overall sensing performance, by providing the most uniform MSE across all target angles.

\appendix
\section{Appendix}

\subsection{Proof of \eqref{eq_CRB_Gamma_APP}}\label{appendix_CRB_largeM_app}
We first derive two useful results regarding the receive steering vector $\mathbf{a}_l\left( \theta\right)$ from \eqref{eq_re_signal2} as following,
\begin{subequations}
    \begin{align}\nonumber
    \mathbf{a}_l\left( \theta\right) = &\exp\left\lbrace j\pi \mathbf{P}^{T}_{\mathrm{S},~l}\mathbf{u}\left(\theta\right)\right\rbrace =\exp\left\lbrace j\pi \left(\mathbf{P}^{\mathrm{L}}_\mathrm{S}-\mathbf{p}^{\mathrm{L}}_{0,~l}\right)^{T}{\mathbf{Q}^{\mathrm{L}}_{l}}\mathbf{u}\left(\theta \right)\right\rbrace\\
    =&\exp\left\lbrace j\pi \mathbf{p}^{\mathrm{L}}_\mathrm{S,~y}\sin\left(\theta-\phi_l \right)+j\pi\frac{M_{\mathrm{S}}\cos\left( \theta-\phi_l\right)}{2\tan\left(\pi/L\right)}\right\rbrace,
    \end{align}
\end{subequations}
and 
\begin{subequations}
    \begin{align} \dot{\mathbf{a}}_l\left( \theta\right)=& \frac{\partial \mathbf{a}_l\left( \theta\right)}{\partial \theta}= j\pi\mathrm{diag}\left\lbrace \mathbf{p}^{\mathrm{L}}_\mathrm{S,~y}\cos\left(\theta-\phi_l\right)-\frac{M_{\mathrm{S}}\sin\left(\theta-\phi_l\right)}{2\tan\left(\pi/L\right)}\right\rbrace{\mathbf{a}_l}\left(\theta\right).    \end{align}
\end{subequations}

In \eqref{eq_CRB_Gamma_APP}, the approximation claims (similarly for $\Gamma_\mathrm{I}(\theta)$)
\begin{subequations}
\begin{align}\label{eq_APPEN_CRB_1_Gamma}
\nonumber
\Gamma_\mathrm{S}(\theta)=&
1-\frac{\|{ \dot{\mathbf{a}}^\mathrm{p}\left(\theta\right)}^H{\mathbf{a}^\mathrm{p}\left(\theta\right)}\|^2}{{ \dot{\mathbf{a}}^\mathrm{p}\left(\theta\right)}^H{ \dot{\mathbf{a}}^\mathrm{p}\left(\theta\right)}{\mathbf{a}^\mathrm{p}\left(\theta\right)}^H{\mathbf{a}^\mathrm{p}\left(\theta\right)}}\\
     \overset{M_\mathrm{S}\uparrow}{=}&1-\frac{\left[\sum_{l}F_{\mathrm{S}}^2(\theta,~l)\sin\left(\theta-\phi_l \right)\right]^2}{\sum_{l}F_{\mathrm{S}}^2(\theta,~l)\sum_{l}F_{\mathrm{S}}^2(\theta,~l)\left[1+\cos^2\left(\theta-\phi_l \right)\tan^2(\pi/L)/3\right]}.
\end{align}
\end{subequations} 
Note that \eqref{eq_APPEN_CRB_1_Gamma} is attained as follows. 

\begin{subequations}\label{eq_APPEN_CRB_colloray}First,
\begin{align} \label{eq_app_a_derivative_a}{\dot{\mathbf{a}}^\mathrm{p}\left(\theta\right)}^H{\mathbf{a}^\mathrm{p}\left(\theta\right)}
    \overset{M\uparrow}{=}\sum_l -j\pi F_\mathrm{S}^2\left(\theta,~l\right)\frac{M_\mathrm{S}^2}{2\tan\left(\pi/L\right)}\sin\left( \theta-\phi_l\right).
    \end{align}
    Second,
     \begin{align}\label{eq_app_a_a}
{\mathbf{a}^\mathrm{p}\left(\theta\right)}^H{\mathbf{a}^\mathrm{p}\left(\theta\right)}=\sum_{l}{F^2_\mathrm{S}\left(\theta,~l\right)}M_\mathrm{S}.  \end{align}
    Third,
     \begin{align}\label{eq_app_a_derivative_a_derivative}
    { \dot{\mathbf{a}}^\mathrm{p}\left(\theta\right)}^H { \dot{\mathbf{a}}^\mathrm{p}\left(\theta\right)}
        \overset{M_\mathrm{S}\uparrow}{=}& \sum_{l}\pi^2{F^2_\mathrm{S}\left(\theta,~l \right) }\left[\frac{\cos^2\left(\theta-\phi_l \right)}{12}+\frac{\sin^2\left(\theta-\phi_l \right)}{4\tan^2\left( \pi/L\right)}\right]M_\mathrm{S}^3.
    \end{align}
\end{subequations}
Finally, combining \eqref{eq_app_a_derivative_a}, \eqref{eq_app_a_a} and \eqref{eq_app_a_derivative_a_derivative}, we have the following relationship,
 \begin{subequations}\label{eq_APPEN_gamma_conclusion}
     \begin{align}\nonumber
&\frac{\|{ \dot{\mathbf{a}}^\mathrm{p}\left(\theta\right)}^H{\mathbf{a}^\mathrm{p}\left(\theta\right)}\|^2}{{ \dot{\mathbf{a}}^\mathrm{p}\left(\theta\right)}^H{ \dot{\mathbf{a}}^\mathrm{p}\left(\theta\right)}{\mathbf{a}^\mathrm{p}\left(\theta\right)}^H{\mathbf{a}^\mathrm{p}\left(\theta\right)}} 
     \\
     \overset{M_\mathrm{S}\uparrow}{=}&\frac{\left(\sum_l  F_\mathrm{S}^2\left(\theta,~l\right)\frac{\sin\left( \theta-\phi_l\right)}{2\tan\left(\pi/L\right)}\right)^2}{\sum_l{F}^2_\mathrm{S}\left(\theta,~l\right) \left( \sum_{l}{F^2_\mathrm{S}\left(\theta,~l \right) }\left[\frac{\cos^2\left(\theta-\phi_l \right)}{12}+\frac{\sin^2\left(\theta-\phi_l \right)}{4\tan^2\left( \pi/L\right)}\right]\right)},\\
     =&1-\Gamma_\mathrm{S}(\theta),
     \end{align}
 \end{subequations}
 which thus completes the proof.

\subsection{Proof of \eqref{eq_CRB_appr} in Proposition 1}\label{appendix_CRB_app}

Based on the assumption of $N_\mathrm{I}=N_\mathrm{S}\triangleq N$ and $\mathbf{F}_{\mathrm{I}}\left(\theta\right)=\mathbf{F}_{\mathrm{S}}\left(\theta\right)\triangleq \mathbf{F}\left(\theta\right)$ in Proposition 1, we have $\mathbf{b}^\mathrm{p}\left(\theta \right)=\mathbf{a}^\mathrm{p}\left(\theta \right)^*$ from \eqref{eq_re_signal2} and $\Gamma_\mathrm{I}(\theta)=\Gamma_\mathrm{S}(\theta)\triangleq \Gamma(\theta)$, which simplifies the CRB in \eqref{eq_CRB_conclusion} as
\begin{subequations}
    \begin{align}\nonumber
      \mathrm{CRB}\left( \theta\right)=&\frac{\sigma^2}{2|\alpha|^2\Gamma\left(\theta \right)}\left[{\mathbf{b}}^\mathrm{p}\left(\theta \right)^T \mathbf{R}_X{\mathbf{b}}^\mathrm{p}\left(\theta \right)^*\dot{\mathbf{a}}^\mathrm{p}\left(\theta \right)^H\dot{\mathbf{a}}^\mathrm{p}\left(\theta \right)
      +\right.\\& \left.\dot{\mathbf{b}}^\mathrm{p}\left(\theta \right)^T\mathbf{R}_X\dot{\mathbf{b}}^\mathrm{p}\left(\theta \right)^*{\mathbf{a}}^\mathrm{p}\left(\theta \right)^H{\mathbf{a}}^\mathrm{p}\left(\theta \right)\right]^{-1}\\\label{eq_APPEN_CRB_APPR1}
      \approx&\frac{\sigma^2}{4|\alpha|^2\Gamma(\theta)}\left[{\mathbf{b}}^\mathrm{p}\left(\theta \right)^T \mathbf{R}_X{\mathbf{b}}^\mathrm{p}\left(\theta \right)^*\dot{\mathbf{a}}^\mathrm{p}\left(\theta \right)^H\dot{\mathbf{a}}^\mathrm{p}\left(\theta \right)\right]^{-1}\\
     \triangleq&\frac{\sigma^2}{4|\alpha|^2\Gamma(\theta)}\left[e \left(\theta \right) r^2 \left(\theta \right) \right]^{-1},
    \end{align}
\end{subequations}
where the approximation in \eqref{eq_APPEN_CRB_APPR1} comes from that,  for large $M$,  $\dot{\mathbf{b}}^\mathrm{p}\left(\theta \right)^T\mathbf{R}_X\dot{\mathbf{b}}^\mathrm{p}\left(\theta \right)^*{\mathbf{a}}^\mathrm{p}\left(\theta \right)^H{\mathbf{a}}^\mathrm{p}\left(\theta \right)$ is asymptotic to ${\mathbf{b}}^\mathrm{p}\left(\theta \right)^T \mathbf{R}_X{\mathbf{b}}^\mathrm{p}\left(\theta \right)^*\dot{\mathbf{a}}^\mathrm{p}\left(\theta \right)^H\dot{\mathbf{a}}^\mathrm{p}\left(\theta \right)$. The proof is given as follows.

Given $\mathbf{b}_l^\mathrm{p}(\theta)=\mathbf{a}_l^\mathrm{p}(\theta)^*$, we have the following two approximations which directly lead to \eqref{eq_APPEN_CRB_APPR1},
\begin{subequations}
    \begin{align}
\dot{\mathbf{b}}^\mathrm{p}\left(\theta \right)^T\mathbf{R}_X\dot{\mathbf{b}}^\mathrm{p}\left(\theta \right)^*=P^{\mathrm{tr}}\sum_{l_1,l_2} {\dot{\mathbf{b}}^p_{l_1}}\left(\theta \right)^T\dot{\mathbf{b}}^p_{l_2}\left(\theta \right)^*\overset{(a)}{\approx}P^{\mathrm{tr}}\dot{\mathbf{a}}^\mathrm{p}\left(\theta \right)^H\dot{\mathbf{a}}^\mathrm{p}\left(\theta \right),
    \end{align}
    and
    \begin{align}{\mathbf{a}}^\mathrm{p}\left(\theta \right)^H{\mathbf{a}}^\mathrm{p}\left(\theta \right)
 \overset{(b)}{\approx} &\sum_{l_1,l_2} {{\mathbf{b}_{l_1}}^p}\left(\theta \right)^H{\mathbf{b}_{l_2}}^p\left(\theta \right)=\frac{1}{P^{\mathrm{tr}}}{{\mathbf{b}}^\mathrm{p}}\left(\theta \right)^H\mathbf{R}_X{\mathbf{b}}^\mathrm{p}\left(\theta \right),
    \end{align}
\end{subequations}
where $(a)$ comes from 
${\dot{\mathbf{a}_{l_1}}^p}\left(\theta \right)^H\dot{\mathbf{a}_{l_2}}^p\left(\theta \right)/\|\dot{\mathbf{a}_{l_1}}^p\left(\theta \right)\|^2 \rightarrow 0$ for $l_1\neq l_2$ for large $M$; $(b)$ comes from 
$\left[{{\mathbf{a}_{l_1}}^p}\left(\theta \right)^H{\mathbf{a}_{l_2}}^p\left(\theta \right)+{{\mathbf{a}_{l_2}}^p}\left(\theta \right)^H{\mathbf{a}_{l_1}}^p\left(\theta \right)\right]/\|{\mathbf{a}_{l_1}}^p\left(\theta \right)\|^2\rightarrow 0$ for $l_1\neq l_2$ for large $M$. For brevity, we prove the approximation in $(b)$ in the following (similar procedures can be obtained for proving $(a)$). Equivalently, the term in $(b)$ has the same decay rate as the following,
\begin{subequations}\label{eq_APP_diff_sector_0}
 \begin{align}
\Big|\frac{{{\mathbf{a}_{l_1}}^p}\left(\theta \right)^H{\mathbf{a}_{l_2}}^p\left(\theta\right)}{\|\mathbf{a}_{l_1}\left(\theta\right)\|^2}\Big|\overset{(c)}{\leq} &\Big|\frac{\exp\left( jM\pi\beta_1\right)\sum_m {\exp\left(j2\pi m\beta_2\right)}}{M}\Big|\\\label{eq_APPEN_prof_decay_rate}
=&\Big|\frac{\sin(M\pi\beta_2) }{M\sin(\pi\beta_2)}\Big|\leq \frac{1}{M\sin(\pi|\beta_2|)},
\end{align}
\text{with}
\begin{align}
\beta_1=&\left[\cos(\theta-\phi_{l_1})-\cos(\theta-\phi_{l_2})\right]/\tan(\pi/L)/2,\\\nonumber
\beta_2=&\left[\sin(\theta-\phi_{l_1})-\sin(\theta-\phi_{l_2})\right]/2\\\label{eq_APPEN_beta21_d}
=&\cos(\frac{2\theta-\phi_{l_1}-\phi_{l_2}}{2})\sin(\frac{\phi_{l_2}-\phi_{l_1}}{2}),
\end{align}
\end{subequations}
where $(c)$ assumes $F_{\mathrm{S}}(\theta,~l_1)\geq F_{\mathrm{S}}(\theta,~l_2)$. \eqref{eq_APPEN_prof_decay_rate} shows a decay rate of $1/M$ as long as the term $\sin(\pi|\beta_2|)$ is well lower-bounded. Therein, we notice that $|\beta_2|$ is lower-bounded by 
\begin{align}\label{eq_APPEN_beta2_d}
|\beta_2|{\geq}& \cos\left( \pi/L\right)\sin\left(\pi/L \right). 
\end{align}   
As an explanation, for an arbitrary pair of $\theta$ and $\phi_{l_1}$, the minimal absolute value (the worst case) of $\beta_2$ is achieved when $\sin(\theta-\phi_{l_2})$ has the same sign as $\sin(\theta-\phi_{l_1})$. In this context, assuming $\phi_{l_2}<\phi_{l_1}\leq \theta$, we have $2\pi/L\leq \phi_{l_1}-\phi_{l_2}\leq \pi-\pi/L$. Also, given  $0 < \theta-\phi_{l_1}\leq \pi-2\pi/L$ and $2\pi/L < \theta-\phi_{l_2}\leq \pi$ (from the condition of half-space antenna radiation), we have  $\pi/L\leq \frac{2\theta-\phi_{l_1}-\phi_{l_2}}{2}\leq \pi-\pi/L$, which results in \eqref{eq_APPEN_beta2_d}. The proof of \eqref{eq_APPEN_prof_decay_rate} and finally \eqref{eq_CRB_appr} are thus completed.

\subsection{The proof of 
 $\Gamma_{\mathrm{I/S}}(\theta)$ in \eqref{eq_CRB_Gamma_APP} fluctuating below $1$}\label{app:Gamma}
Equivalently, we show $\frac{\|{\dot{\mathbf{a}}^\mathrm{p}(\theta)}^H{\mathbf{a}^\mathrm{p}(\theta)}\|^2}{{\dot{\mathbf{a}}^\mathrm{p}(\theta)}^H{\dot{\mathbf{a}}^\mathrm{p}(\theta)}{\mathbf{a}^\mathrm{p}(\theta)}^H{\mathbf{a}^\mathrm{p}(\theta)}}$ in \eqref{eq_APPEN_gamma_conclusion} fluctuates above $0$.  Note that the following derivations can be readily extended to a larger $L$ with a higher accuracy. 

\subsubsection{For $L=2$}  we notice that ${\dot{\mathbf{a}}^\mathrm{p}(\theta)}^H{\mathbf{a}^\mathrm{p}(\theta)}=0$ always holds from \eqref{eq_CRB_Gamma_APP} since $\tan(\pi/2)=\infty$, which completes the proof.

\subsubsection{For $L=3$} we consider a periodicity of $\theta\in \left(0, ~2\pi/3\right)$, and have the following results. Firstly, the numerator in \eqref{eq_APPEN_gamma_conclusion} is further expressed as
   \begin{subequations}
    \begin{align} &\|{\dot{\mathbf{a}}^\mathrm{p}(\theta)}^H{\mathbf{a}^\mathrm{p}(\theta)}\|^2\overset{M\uparrow}{\approx}\frac{\pi^2M^4}{12}\left(\sum_l  F^2\left(\theta,~l\right)\sin\left( \theta-\phi_l\right) \right)^2\\
    =&\begin{cases}
        \frac{\pi^2M^4}{12}\left(  F^2\left(\theta,~3\right)\sin\left( \theta-\Phi_3\right) \right)^2,~&\theta\in \left(0, ~\pi/3\right),\\
        \frac{\pi^2M^4}{12}\left(\sum_{l=1,3}  F^2\left(\theta,~l\right)\sin\left( \theta-\Phi_l\right) \right)^2,~&\theta\in \left(\pi/3, ~2\pi/3\right)
    \end{cases}\\
    \leq &\begin{cases}
        \frac{\pi^2M^4}{48}  F^4\left(\theta,~3\right),~&\theta\in \left(0, ~\pi/3\right),\\
        \frac{\pi^2M^4}{48}\max\left( F^4\left(\theta,~1\right),~F^4\left(\theta,~3\right)\right),~&\theta\in \left(\pi/3, ~2\pi/3\right)
\end{cases}\\\label{eq_app_a_d_a_L3}
    \approx & \begin{cases}
    \frac{\pi^2M^4}{48}  F^4\left(\theta,~3\right),~&\theta\in \left(0, ~\pi/3\right),\\
        \frac{\pi^2M^4}{48}\max\left( F^4\left(\theta,~1\right),~F^4\left(\theta,~3\right)\right)~&\theta\in \left(\pi/3, ~2\pi/3\right).
    \end{cases}
    \end{align}
\end{subequations}
Similarly, the denominator in \eqref{eq_APPEN_gamma_conclusion} is further expressed as
\begin{subequations}
    \begin{align}\nonumber
&{\dot{\mathbf{a}}^\mathrm{p}(\theta)}^H{\dot{\mathbf{a}}^\mathrm{p}(\theta)}{\mathbf{a}^\mathrm{p}(\theta)}^H{\mathbf{a}^\mathrm{p}(\theta)}\\
=&\frac{\pi^2M^2}{12}\sum_lF^2\left(\theta,~l\right)\left[M^2-\cos^2\left(\theta-\phi_l \right)\right]\sum_lF^2\left(\theta,~l\right)\\\label{eq_app_a_d_a_d_aa_L3}
\approx&\frac{\pi^2M^4}{12}\left[\sum_lF^2\left(\theta,~l\right)\right]^2\\
\geq& \begin{cases}
    \frac{\pi^2M^4}{12}  F^4\left(\theta,~3\right),~&\theta\in \left(0, ~\pi/3\right),\\
        \frac{\pi^2M^4}{12}\max\left( F^4\left(\theta,~1\right),~F^4\left(\theta,~3\right)\right)~&\theta\in \left(\pi/3, ~2\pi/3\right).
\end{cases}
        \end{align}
\end{subequations}
    Combining \eqref{eq_app_a_d_a_L3} and \eqref{eq_app_a_d_a_d_aa_L3}, we have $\frac{\|{\dot{\mathbf{a}}^\mathrm{p}(\theta)}^H{\mathbf{a}^\mathrm{p}(\theta)}\|^2}{{\dot{\mathbf{a}}^\mathrm{p}(\theta)}^H{\dot{\mathbf{a}}^\mathrm{p}(\theta)}{\mathbf{a}^\mathrm{p}(\theta)}^H{\mathbf{a}^\mathrm{p}(\theta)}}\leq 1/4$.

\subsubsection{For $L=4$} we assume the target is illuminated by sector $1$ and sector $4$. For the isotropic antenna pattern, we further express the numerator in \eqref{eq_APPEN_gamma_conclusion} as $\|{\dot{\mathbf{a}}^\mathrm{p}(\theta)}^H{\mathbf{a}^\mathrm{p}(\theta)}\|^2\approx 8M^4\pi^2\sum_l \left[ \frac{\cos^2\left(\theta-\phi_l \right)}{12}+\frac{\sin^2\left(\theta-\phi_l \right)}{4}\right]=\frac{8M^4\pi^2}{3}$, which gives $\frac{\|{\dot{\mathbf{a}}^\mathrm{p}(\theta)}^H{\mathbf{a}^\mathrm{p}(\theta)}\|^2}{{\dot{\mathbf{a}}^\mathrm{p}(\theta)}^H{\dot{\mathbf{a}}^\mathrm{p}(\theta)}{\mathbf{a}^\mathrm{p}(\theta)}^H{\mathbf{a}^\mathrm{p}(\theta)}}\leq 3/8$. For the directive antenna pattern in \eqref{eq_antenna_pattern_dir}, the numerator in \eqref{eq_APPEN_gamma_conclusion} is expressed as $\|{\dot{\mathbf{a}}^\mathrm{p}(\theta)}^H{\mathbf{a}^\mathrm{p}(\theta)}\|^2 \approx \frac{3\pi^2M^4}{8}\left[\sum_l \sin\left(2\theta-\phi_l\right)\right]^2\leq \frac{3\pi^2M^4}{4}$  and the denominator in \eqref{eq_APPEN_gamma_conclusion} is expressed as ${\dot{\mathbf{a}}^\mathrm{p}(\theta)}^H{\dot{\mathbf{a}}^\mathrm{p}(\theta)}{\mathbf{a}^\mathrm{p}(\theta)}^H{\mathbf{a}^\mathrm{p}(\theta)}\approx 3M^4\pi^2\sum_l \left[2 \cos^2\left(\theta-\phi_l\right)\sin^2\left(\theta-\phi_l\right)+{\cos^2\left(\theta-\phi_l \right)}\right]\geq 3M^4\pi^2$, which gives $\frac{\|{\dot{\mathbf{a}}^\mathrm{p}(\theta)}^H{\mathbf{a}^\mathrm{p}(\theta)}\|^2}{{\dot{\mathbf{a}}^\mathrm{p}(\theta)}^H{\dot{\mathbf{a}}^\mathrm{p}(\theta)}{\mathbf{a}^\mathrm{p}(\theta)}^H{\mathbf{a}^\mathrm{p}(\theta)}}\leq 1/4$.

Note that from our simulations in Fig. \ref{fig_diff_SNR}, the fluctuations of $\Gamma_{\mathrm{I/S}}(\theta)$ are also shown to be very trivial and below 1.

\subsection{Gain of specific antenna patterns}\label{app:AP derive}
{The gain of a half-space isotropic antenna pattern (after power normalization) is typically written as  }
\begin{align}
\label{app_eq_antenna_pattern_iso3}  
G^{\mathrm{Iso}}\left(\Phi\right)=&   \begin{cases}
2,&0\leq \Phi  \leq \pi/2,\\
0,&\text{otherwise,}
\end{cases}
\end{align}
{where $0\leq \Phi\leq \pi$ is the elevation angle in a newly defined CCS,  denoted as $x_\mathrm{a}-y_\mathrm{a}-z_\mathrm{a}$, whose $x_\mathrm{a}-y_\mathrm{a}$ plane is aligned with the antenna plane.  The $\Phi$ in $x_\mathrm{a}-y_\mathrm{a}-z_\mathrm{a}$ is related to the $\theta_l$ in $x_l-y_l$  since $\Phi=\theta_l$ for $0\leq \theta_l\leq \pi$, and $\Phi=\theta_l$ for $\Phi=-\theta_l$ for $-\pi\leq \theta_l\leq 0$. Hence, $\cos(\Phi)=\cos(\theta_l)$ and thus  $0\leq \Phi\leq \pi/2$  is equivalent to $\cos(\theta_l) \geq 0$.} As a result, \eqref{app_eq_antenna_pattern_iso3} is modified as
\begin{align}
\label{app_eq_antenna_pattern_iso}  
G^{\mathrm{Iso}}\left(\theta_l\right)=&   \begin{cases}
2,& \cos(\theta_l)\geq 0,\\
0,&  \text{otherwise,}
\end{cases}
\end{align}
 {in L-CCS $x_l-y_l  $}.

A similar conversion can be made to obtain the gain of the directive antenna pattern in \eqref{eq_antenna_pattern_dir} \cite{li2023beyond}.
\eqref{eq_antenna_pattern_dir} can be explicitly expressed by
\begin{subequations}\label{eq_appendix_appr_F}
 \begin{align}  
   F\left(\theta,~l\right)=&\sqrt{ G^{\mathrm{Dir}}\left(\theta_l\right)}=
\sqrt{2\left(\alpha_{\mathrm{L}}+1\right)}\cos^{\alpha_\mathrm{L}/2}\left(\theta_l\right)\\
=&\sqrt{2\left(\alpha_{\mathrm{L}}+1\right)}\cos^{\alpha_\mathrm{L}/2}\left[\arccos\left(\frac{[\mathbf{Q}^T_l\left( \mathbf{p}_{\mathrm{T}}-\mathbf{p}_{\mathrm{c},~l}\right)]_1}{\|\mathbf{Q}^T_l\left( \mathbf{p}_{\mathrm{T}}-\mathbf{p}_{\mathrm{c},~l}\right)\|_2}\right)\right]\\
\approx &\sqrt{2\left(\alpha_{\mathrm{L}}+1\right)}   \left(\frac{[\mathbf{Q}^T_l\mathbf{p}_{\mathrm{T}}]_1}{\|\mathbf{Q}^T_l\mathbf{p}_{\mathrm{T}}\|_2}\right) ^{\alpha_\mathrm{L}/2}\\
=&\sqrt{2\left(\alpha_{\mathrm{L}}+1\right)}   \cos\left(\theta-\phi_l\right) ^{\alpha_\mathrm{L}/2},~\text{for }-\pi/2
\leq \left(\theta-\phi_l\right)\leq \pi/2,
\end{align}   
\end{subequations}
where $\alpha_\mathrm{L}=0,~1,~2$ for $L=2,~3,~4$, respectively.

\subsection{The scaling terms in Table \ref{tab2}, given $L=3,~4$ with directive  antenna patterns}\label{Appendix_e_r}
We only provide the scaling laws for $L=3,~4$ for conciseness. Other scenarios can be obtained with a straightforward extension.
\subsubsection{For $L=3$} the expected scaling law given the isotropic antenna pattern is straightforward. For the directive antenna pattern, if $\theta$ ranges from $0$ to $\pi/3$, we have  the following approximation
\begin{subequations}
 \begin{align}
    \mathbb{E}\left\lbrace r^2\left(\theta\right)\right\rbrace&\overset{N~\uparrow}{\approx}\frac{3}{\pi}\int_{0}^{\pi/3} \left\lbrace \frac{\pi^2N^3\cos\left({\theta}-\phi_3 \right)}{81} \right\rbrace d \theta\\    &\overset{N~\uparrow}{=}\frac{3}{\pi}\frac{\pi^2N^3}{81}\sin\left({\theta}\right)\Big|_{-\pi/6}^{\pi/6}\approx 0.116N^3.
\end{align}   
\end{subequations}

If $\theta$ ranges from $\pi/3$ to $2\pi/3$, we have  the following approximation
\begin{subequations}
  \begin{align}
    \mathbb{E}\left\lbrace r^2\left(\theta\right)\right\rbrace
    \overset{N~\uparrow}{\approx}&\frac{3}{2\pi}\int_{0}^{\pi/3} \left\lbrace \frac{\pi^2N^3\left[\cos\left({\theta}-\phi_3 \right)+\cos\left({\theta}-\phi_1 \right)\right]}{81} \right\rbrace d \theta\\
    =&\frac{3}{\pi}\int_{\pi/3}^{2\pi/3} \frac{\pi^2N^3 \sin\left({\theta}\right)}{81} d {\theta}\approx 0.116N^3.
\end{align}  
\end{subequations}

\subsubsection{For $L=4$} each target is illuminated by 2 aligned sectors (assumed to be sector $1$ and $4$). Given directive   antenna patterns, we have the following approximation
\begin{subequations}
   \begin{align}
    \label{eq_app_rate_L4N_scaling}
    r^2\left(\theta \right)|_{L=4}
    \overset{N~\uparrow}{\approx}&\sum_{l=1,4}\frac{N^3\pi^2{F\left(\theta,~l \right) }^2}{768}\left[2\sin^2\left(\Theta_l \right)+1\right]\\\nonumber
    \overset{(a)}{=}
&\frac{\pi^2N^3 }{128}\sin^2\left(2\theta-2\bar{\Phi}-\frac{\pi}{2} \right)
    +\frac{\pi^2N^3 }{128}\\
    =&\frac{3-\cos(4\theta)}{128} \pi^2N^3,
\end{align} 
\end{subequations}
where $\Theta_l\triangleq\theta-\phi_l$. $(a)$ defines $\bar{\Phi}\triangleq(\Phi_{1}+\Phi_{4})/2$ with  specifically $\Phi_4=-\pi/4$ and $\Phi_1=\pi/4$.

The expectation of $r^2\left(\theta \right)|_{L=4}$ over $\theta$ is then written as
\begin{subequations}
    \begin{align}
    \mathbb{E}\left\lbrace r^2\left(\theta\right)\right\rbrace\overset{N~\uparrow}{\approx}&\frac{2}{\pi}\int_{-\pi/4}^{\pi/4} \left\lbrace \frac{\cos^2\left(2\theta \right)+1}{128}\pi^2N^3\right\rbrace d \theta=0.116N^3.
\end{align}
\end{subequations}

\bibliographystyle{IEEEtran}
\bibliography{IEEEabrv, references}
 \end{document}